\documentclass[prd,aps,amsfonts,showpacs,longbibliography,notitlepage,twocolumn,groupedaddress,superscriptaddress,nofootinbib,10pt]{revtex4-2}

\usepackage[utf8]{inputenc}

\usepackage{tikz} 
\usepackage{amsmath,amsfonts,amsthm, mathtools}
\usepackage{placeins}
\usepackage{float}
\usepackage{dsfont}
\usepackage{csquotes}
\usepackage{amssymb}
\usepackage{verbatim}
\usepackage{bm}
\usepackage{amsmath}
\usepackage{amssymb}
\usepackage{stmaryrd}
\usepackage{amsthm}
\usepackage[caption=false]{subfig}
\usepackage{physics}
\usepackage{bm}  

\usepackage{bibunits}

\usepackage{hyperref}
\definecolor{darkred}  {rgb}{0.5,0,0}
\definecolor{darkblue} {rgb}{0,0,0.5}
\definecolor{darkgreen}{rgb}{0,0.5,0}
\hypersetup{
  colorlinks = true,
  urlcolor  = blue,         
  linkcolor = darkblue,     
  citecolor = darkgreen,    
  filecolor = darkred       
}

\theoremstyle{definition}

\newtheorem{corollary}{Corollary}
\newtheorem{definition}{Definition}

\newtheorem{lemma}{Lemma}
\newtheorem{proposition}{Proposition}
\newtheorem{thm}{Theorem}

\newcommand{\mbb}{\mathbb}
\newcommand{\mc}{\mathcal}

\newcommand{\wt}{\widetilde}

\usepackage{multirow}

\usepackage{xparse}

\DeclareDocumentCommand{\norm}{m}{\lVert #1\rVert}

\newcommand{\x}{\mathbf{x}}

\definecolor{cool_green}{rgb}{0.0, 0.5, 0.0}

\newcommand{\yk}[1]{{ #1}}
\definecolor{gold}{rgb}{0.85, 0.65, 0.13}
\usepackage[dvipsnames]{xcolor}

\begin{document}

\title{Private Correlations Certify Sensing Capability}
\author{Yunkai Wang}
\email{ywang10@perimeterinstitute.ca}
\affiliation{Perimeter Institute for Theoretical Physics, Waterloo, Ontario N2L 2Y5, Canada.}
\affiliation{Department of Applied Mathematics, University of Waterloo, Ontario N2L 3G1, Canada.}
\affiliation{Institute for Quantum Computing, University of Waterloo, Ontario N2L 3G1, Canada.}
\affiliation{Department of Physics and Astronomy, University of Waterloo, Ontario N2L 3G1, Canada.}

\author{Peixue Wu}
\affiliation{Department of Applied Mathematics, University of Waterloo, Ontario N2L 3G1, Canada.}
\affiliation{Institute for Quantum Computing, University of Waterloo, Ontario N2L 3G1, Canada.}

\author{Graeme Smith}
\affiliation{Department of Applied Mathematics, University of Waterloo, Ontario N2L 3G1, Canada.}
\affiliation{Institute for Quantum Computing, University of Waterloo, Ontario N2L 3G1, Canada.}

\author{Sisi Zhou}
\affiliation{Perimeter Institute for Theoretical Physics, Waterloo, Ontario N2L 2Y5, Canada.}
\affiliation{Department of Applied Mathematics, University of Waterloo, Ontario N2L 3G1, Canada.}
\affiliation{Institute for Quantum Computing, University of Waterloo, Ontario N2L 3G1, Canada.}
\affiliation{Department of Physics and Astronomy, University of Waterloo, Ontario N2L 3G1, Canada.}

\begin{abstract}
We show that private correlations in a bipartite quantum state constitute a metrological resource for distributed sensing assisted by a possibly noisy channel from one party to the other. We begin with an example with distillable secret key but poor locally accessible sensing performance and show that an assisting channel substantially improves its performance. \yk{We then show that any shared state yielding positive private information can serve as a sensing resource after appropriate local coherent processing, and derive a quantitative lower bound on the locally accessible sensing capability after assistance.
We further show that classical correlations alone provide no analogous guarantee.} Finally, in a channel formulation, we show that a channel with positive private information allows nonzero locally accessible sensitivity when paired with a suitable assisting channel.
\end{abstract}



\maketitle


\textit{Introduction} -- Quantum sensing exploits fundamentally quantum strategies to estimate the parameter dependence of quantum states to achieve sensitivities beyond those attainable with classical resources~\cite{BraunsteinCaves1994,giovannetti2004quantum,giovannetti2006quantum,paris2009quantum,GiovannettiLloydMaccone2011,degen2017quantum,pezze2018quantum,braun2018quantum,pirandola2018advances,proctor2018multiparameter,pezze2009entanglement}.
\yk{However, quantum states are subject to noise and loss during preparation, distribution, and storage. These imperfections can substantially reduce their fidelity with the intended state.}
General frameworks have been developed that give upper bounds on the achievable precision under noisy conditions~\cite{EscherDavidovich2011,DemkowiczDobrzanski2012}.
Here we pursue an important complementary problem: can an operational information-theoretic quantity of a noisy quantum state certify a nontrivial lower bound on its sensing capability? Drawing on insights from quantum Shannon theory, we identify one such certificate, offering a new perspective on translating the operational properties of noisy quantum states into rigorous metrological guarantees.

Quantum Shannon theory provides a natural hierarchy of operational quantities that characterize which informational resources survive noise \cite{devetak2005private,devetak2005distillation,horodecki2009general}.
As the noise becomes stronger, distillable entanglement may vanish while private correlations persist, allowing two parties to distill a shared secret key \cite{horodecki2005secure}. 
\yk{This demonstrates that privacy occupies a weaker, and therefore more noise-robust, level of the hierarchy than distillable entanglement. }
Distillable entanglement and distillable key capture quantitatively these two distinct resources. 
Entanglement has long been regarded as a central resource for quantum sensing, and previous work has established that even weak forms of entanglement, including PPT bound entanglement, can provide a metrological advantage over separable probes \cite{czekaj2015quantum,toth2018quantum,toth2020activating,pal2021bound}. These results considerably broaden the classes of entangled states known to be useful for sensing, while continuing to interpret their metrological advantage primarily through the lens of entanglement. 
\yk{Complementary characterizations of sensing resources have been developed using generator-dependent coherence \cite{streltsov2017colloquium} and asymmetry \cite{kudo2023fisher}, as well as discord-type measures \cite{girolami2014quantum}.}
More generally, every resourceful state is advantageous in a suitably constructed estimation task \cite{tan2021fisher}, but this result does not provide a nontrivial lower bound for a preassigned sensing model.
It remains unknown whether private correlations can explicitly certify metrological capability; prior work on privacy in sensing instead addresses
protocol security
\cite{hassani2025privacy,namkung2026universal,alushi2026privacy,
de2025anonymous,bugalho2025private,wang2025secure}. 


%


\yk{In this work, we establish that private correlations provide a quantitative certificate of distributed sensing capability by guaranteeing surviving phase-sensitive coherence that can be made locally accessible with an assisting channel.}
We first demonstrate this mechanism using a family of private states constructed from Werner data-hiding states \cite{terhal2001hiding,divincenzo2002quantum,horodecki2005secure,matthews2009chernoff,matthews2009distinguishability,christandl2017private}. Without the assisting channel, the local operations and classical communication (LOCC) Fisher information is suppressed inversely with the dimension of the system. After Alice transmits her  parameter-independent subsystem to Bob, the LOCC Fisher information is much improved and becomes independent of the dimension of the  system.
Remarkably, this recovery persists when the assisting channel is an erasure or depolarizing channel with any fixed nonzero transmission strength.


We then establish a general relation between privacy and sensing. Consider a general state shared by Alice and Bob, with the parameter encoded on one of Alice's subsystems. If measuring this subsystem yields positive private information between Alice and Bob, then transmitting Alice's remaining unencoded subsystem through an assisting channel which can be noisy guarantees a strictly positive quantitative lower bound on the LOCC Fisher information. Essentially, the assisting channel does not carry information about the unknown parameter. 
\yk{We show that any shared state yielding positive private information can be converted through local coherent processing into this resource form.
We also show that classical correlations provide no analogous guarantee.} We further develop a channel-based sensing model, showing that positive private information of the channel likewise certifies LOCC-accessible sensing capability when the purifying subsystem of the input ensemble is transmitted through an assisting channel. 

Our results thus identify a broad class of states that can be useful for quantum sensing based on their information-theoretic properties, providing general guidance for the study of noisy quantum sensing. More generally, we connect privacy properties with guaranteed metrological performance, establishing a bridge between the resources that survive noise and the sensing capabilities they can support. This connection motivates further investigation into whether the well-developed tools and novel phenomena of quantum communication can yield new insights and protocols for quantum sensing.

\begin{figure}[!tb]
\begin{center}
\includegraphics[width=0.68\columnwidth]{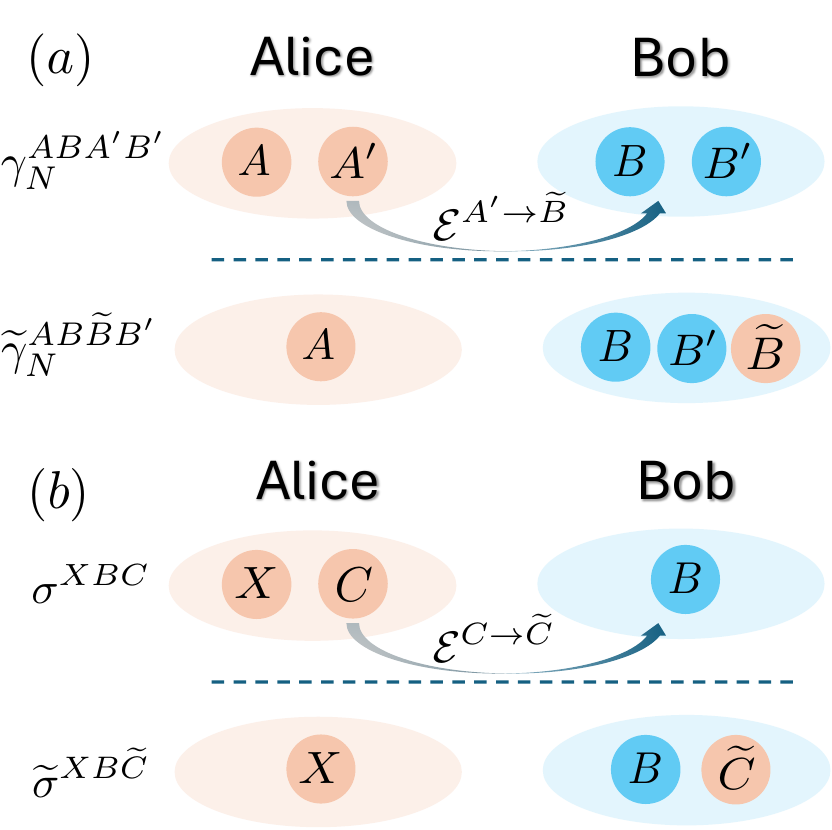}
\caption{Three formulations in which private correlations certify distributed sensing capability: (a) the data-hiding example, (b) the general-state formulation. 
}
\label{fig:Set_up}
\end{center}
\end{figure}

\textit{Example constructed with data hiding states}\label{Sec:data_hiding_distributed} -- We begin with an example showing that a state containing private correlations can exhibit a substantial enhancement in its locally accessible metrological capability when supplemented by an assisting channel. As illustrated in Fig.~\ref{fig:Set_up}(a), consider the following family of states constructed from Werner data-hiding states \cite{terhal2001hiding,divincenzo2002quantum,horodecki2005secure,matthews2009chernoff,matthews2009distinguishability,christandl2017private}:
\begin{equation}
\begin{aligned}
\label{eq:general-hidden-sensing-state}
    &\gamma_N^{ABA'B'}(\theta)=\frac12\,\psi_{N+}^{AB}(\theta)\otimes \rho_+^{A'B'}+\frac12\,\psi_{N-}^{AB}(\theta)\otimes \rho_-^{A'B'},\\
    &\rho_+^{A'B'}=\frac{I+S}{d(d+1)},
    \quad
    \rho_-^{A'B'}=\frac{I-S}{d(d-1)} ,   
\end{aligned}
\end{equation}
where $\ket{\psi_{N\pm}(\theta)}=\frac{1}{\sqrt2}\bigl(\ket{N0}\pm e^{iN\theta}\ket{0N}\bigr)$ and $\psi_{N\pm}^{AB}(\theta)=\ket{\psi_{N\pm}(\theta)}\bra{\psi_{N\pm}(\theta)}^{AB}$, $I$ is the identity matrix, $S$ is the swap operator. 
\yk{Here, $N$ denotes the total excitation number. In a photonic realization, $|N0\rangle$ and $|0N\rangle$ are Fock states with all $N$ photons in one of the two modes.} The states $\rho_+^{A'B'}$ and $\rho_-^{A'B'}$ are the Werner data-hiding states on $A'B'$. Initially, Alice holds the systems $AA'$, while Bob holds $BB'$. Alice may send the system $A'$ to Bob through an assisting channel $\mc E^{A'\to \wt B}$, and we define
\begin{equation}\label{eq:general-noisy-hidden-sensing-state}
    \wt\gamma_N^{AB\wt B B'}(\theta)
    :=(\mc E^{A'\to \wt B}\otimes \mathrm{id}^{ABB'})(\gamma_N^{AB A' B'}(\theta)).
\end{equation}
Thus, $\gamma_N(\theta)$ and $\wt\gamma_N(\theta)$ denote the sensing states before and after Alice transmits her unencoded subsystem through $\mc E$, respectively.

Eq.~\eqref{eq:general-hidden-sensing-state} contains one private bit: measuring $AB$ in the basis ${\lvert N0\rangle,\lvert 0N\rangle}$ yields a perfectly correlated bit independent of the purifying environment. This state is inspired by the construction that rigorously separated privacy from distillable entanglement~\cite{horodecki2005secure} and enabled the discovery of superactivation~\cite{smith2008quantum}. The data-hiding states $\rho_\pm^{A'B'}$ are perfectly distinguishable globally, whereas their LOCC distinguishability vanishes as $d\to\infty$. The phase-encoded states $\psi_{N\pm}^{AB}$ differ by a phase flip whose label is stored in $\rho_\pm^{A'B'}$; inferring this label therefore requires distinguishing  $\rho_\pm^{A'B'}$, which LOCC cannot do reliably for large $d$. After Alice sends $A'$ through the assisting channel, $A$ and $B$ remain separated, but Bob can jointly measure the received subsystem $A'$ and $B'$. Even if the assisting channel is noisy, the resulting shield states can retain sufficient distinguishability for Bob to infer the label and condition the subsequent LOCC sensing measurement. We therefore expect assistance to improve locally accessible sensing, although it carries neither the encoded system nor the parameter. The following theorem confirms this intuition:


\begin{thm}\label{main_thm:hiding-recovery}
For LOCC measurement on the state $\gamma_N$ before using the assisting channel, Fisher information is upper bounded
\begin{equation}
    F_{\mathrm{LOCC}}^{(AA':BB')}(\gamma_N(\theta))
    \le
    \frac{4N^2}{d+1}.
\end{equation}
Use the erasure channel  $\mc E_\lambda(X)=(1-\lambda)X+\lambda\,\Tr(X)\ket e\bra e$ as the assisting channel, 
\begin{equation}
    F_{\mathrm{LOCC}}^{(A:\wt B B B')}(\wt \gamma_N(\theta))
    \ge
    (1-\lambda)N^2.
\end{equation}
Use the depolarizing channel $\mc D_p(X)=(1-p)X+\frac{p}{d}\Tr(X)\,I$ as assisting channel, 
\begin{equation}
    F_{\mathrm{LOCC}}^{(A:\wt B B B')}(\wt \gamma_N(\theta))
    \ge
    \frac{(1-p)^2}{1-p^2/d^2}\,N^2.
\end{equation}
\end{thm}

In the limit $d\gg 1$, the above bounds establish a separation between the LOCC sensing performance with and without assistance that grows at least linearly with $d$.
Importantly, this enhancement does not require a noiseless assisting channel: for any fixed $\lambda<1$ or $p<1$, the assisted LOCC Fisher information remains of order $N^2$. 
\yk{Thus, we introduce a new paradigm of noisy distributed sensing: starting from a noisy state shared by Alice and Bob, Alice transmits a subsystem to Bob through a noisy channel to assist the sensing. The phase-flip uncertainty in the sensing state is correlated with the shield, whose relevant information is difficult to extract by LOCC. The distributed resource is itself noisy and can have arbitrarily little distillable entanglement, yet assistance makes its phase information locally accessible. }
Because the assisting channel transmits only the parameter-independent subsystem, the parameter-encoded systems remain distributed between Alice and Bob. Parameter estimation must therefore still be performed through an LOCC protocol involving local measurements by both parties.
\yk{Notably, $\gamma_N$ has distillable key $K_D\geq 1$ but distillable entanglement $E_D=O(1/d)$ \cite{vidal2002computable,plenio2005logarithmic,horodecki2005secure}. This separation parallels the LOCC sensing gap with and without assistance, suggesting that privacy can certify assisted sensing even as the distillable entanglement tends to zero.
} Further discussion and detailed proofs are provided in Sec. B of the Supplemental Material.



\textit{States with private correlation} -- We now generalize the above example and show that privacy certifies sensing capability for general states shared by Alice and Bob. More precisely, if measuring one of Alice's subsystems in the encoding basis yields positive private information with Bob, then an assisting-channel recovery strategy guarantees a lower bound on the LOCC Fisher information.



As illustrated in Fig.~\ref{fig:Set_up}(b), let $\sigma^{XBC}_\theta=U^X_\theta\sigma^{XBC}_0
U^{X\dagger}_\theta$ be a state shared by Alice and Bob, where Alice holds $XC$ and Bob holds $B$. For an arbitrary encoding $U^X_\theta=e^{i\theta H_X}$ on system $X$, write the spectral decomposition as $H_X=\sum_x h_x\ket{x}\bra{x}^X$. \yk{We assume $H_X$ is non-degenerate and allow it to act jointly on several degrees of freedom within $X$.} We then express $\sigma^{XBC}_0$ in the corresponding $X$-block form as
\begin{equation}
\sigma^{XBC}_0
=
\sum_{x,x'}
|x\rangle\!\langle x'|^X
\otimes
\sigma^{BC}_{xx'} .
\end{equation}
\yk{The state $\sigma_0^{XBC}$ describes the noisy resource available after preparation, distribution, and storage, and already incorporates potential noise, loss, and other imperfections. We ask: if private correlations nevertheless survive these noisy processes, can they certify sensing capability? }
Define $p_x=\operatorname{Tr}\sigma^{BC}_{xx}$, $X_+=\{x:p_x>0\}$, and $m=|X_+|$. After measuring $X$ in this basis and discarding Alice's remaining system $C$, the resulting cq state shared by Alice, Bob, and the purifying environment $E$, is
\begin{equation}
\omega_{\sigma_0}^{XBE}
:=
\operatorname{Tr}_C\!\left[
\Delta_X\!\left(|\Omega\rangle\langle\Omega|^{XCBE}\right)
\right],    
\end{equation}
where $\ket{\Omega}^{XCBE}=\sum_x\sqrt{p_x}\ket{x}^X\ket{\phi_x}^{BCE}$ is any purification of $\sigma_0^{XBC}$, and $\Delta_X(\cdot)=\sum_x |x\rangle\langle x|^X(\cdot)|x\rangle\langle x|^X$ denotes the dephasing map. We define the fixed-$X$ private information by
\begin{equation}
 P_X(\sigma_0)
:=
I(X;B)_{\omega_{\sigma_0}}
-
I(X;E)_{\omega_{\sigma_0}} .   
\end{equation}
\yk{By the Devetak--Winter bound, $P_X(\sigma_0)$ lower bounds the asymptotic secret-key rate achievable from the postmeasurement state $\omega_{\sigma_0}^{XBE}$ \cite{devetak2005distillation}. Thus, $P_X(\sigma_0)>0$ is a standard sufficient condition for private-key generation.}

\yk{We first explain why private information is a natural certificate of sensing capability.
Under the encoding $U_\theta^X =e^{i\theta H_X}$, each off-diagonal block $\sigma^{BC}_{xx'}$ ($x\neq x'$) acquires a relative phase and thus carries the sensing signal. Hence, nonzero Fisher information requires some of this coherence to survive and be accessible to the measurement. 
Positive private information certifies the survival of this phase-sensitive off-diagonal coherence, as follows from
\begin{equation}
\frac{\|\sigma^{BC}_{xx'}\|_1}{\sqrt{p_xp_{x'}}}=\|\tr_E(\ket{\phi_x}\bra{\phi_{x'}}^{BCE})\|_1=F(\sigma_x^E,\sigma_{x'}^E). 
\end{equation}
where
$\sigma_x^E=\tr_{BC}(\ket{\phi_x}\bra{\phi_x}^{BCE})$, and $F(\cdot,\cdot)$ denotes the fidelity. 
Indistinguishability of the environmental states is equivalent to nonzero off-diagonal coherence in the $X$-block decomposition, and positive private information ensures that these environmental states cannot be perfectly distinguished.}
This intuition is captured quantitatively by specializing known bounds from the resource theory of coherence to the encoding-basis block decomposition considered here~\cite{rana2017logarithmic,bischof2021quantifying,chitambar2019quantum}. 
Importantly, converting the coherence certified by the state's privacy into an LOCC-accessible sensing signal requires Bob to access both $B$ and $C$. When $B$ and $C$ remain separated, this coherence may not be locally accessible, as illustrated by the example in Theorem~\ref{main_thm:hiding-recovery}. Transmitting $C$ to Bob enables a joint $BC$ readout whenever the transmission succeeds.
By combining the privacy--coherence connection with an explicit protocol using an assisting channel followed by a constructive LOCC measurement, we convert private information into the quantitative guarantee on locally accessible Fisher information stated in the following theorem.


\begin{thm}\label{thm:general-fixed-key-assisted}
Let $\sigma^{XBC}_\theta$ be the state defined above, with Alice holding $XC$ and Bob holding $B$. 
Assume that $P_X(\sigma_0)>0$. After sending $C$ through an erasure channel with erasure probability $\lambda$,
\begin{equation}
\widetilde{\sigma}^{XB\widetilde C}_\theta
:=
(\operatorname{id}_{XB}\otimes
\mathcal E^{C\to\widetilde C}_{\lambda})
(\sigma^{XBC}_\theta),    
\end{equation}
there exists a LOCC measurement such that
\begin{equation}
\left.
F_{\mathrm{LOCC}}^{(X: B\widetilde C)}
\left(
\widetilde{\sigma}^{XB\widetilde C}_\theta
\right)
\right|_{\theta=\theta_0}
\ge
(1-\lambda)
\frac{\delta_H^2}{4L_m^2}
\left(2^{P_X(\sigma_0)}-1\right)^2.
\end{equation}
where $\delta_H
:=
\min_{\substack{x\neq x'\\x,x'\in X_+}}
|h_x-h_{x'}|>0$,
$L_m
:=
\max_{1\leq j\leq m}
\left(
\sum_{\ell=1}^{j-1}\frac{1}{\ell}
+
\sum_{\ell=1}^{m-j}\frac{1}{\ell}
\right)=\Theta(\log m)$.
\end{thm}
Because the assisting transmission is parameter independent, it cannot increase the global quantum Fisher information; it only unlocks sensitivity that was inaccessible under the original LOCC partition.
\yk{The above theorem establishes a sufficient, but not necessary, condition for nonzero assisted LOCC Fisher information.}
Further discussion and detailed proofs are provided in Sec. C of the Supplemental Material.


In the small-$m$ regime, as exemplified by Theorem~\ref{main_thm:hiding-recovery}, private information provides a constant lower bound on the LOCC Fisher information. 
In the large-$m$ regime, it is useful to make the scaling with $m$ more explicit. Consider the situation where $\norm{H_X}_\infty = 1$ and the eigenvalues are equally spaced between $-1$ and $1$ so that $\delta_H = \Theta(1/m)$. Also, assume the fixed-$X$ private information is nearly maximal, i.e. $P_X=\log m-O(1)$\footnote{Throughout this work, all logarithms are taken to base two.}. Theorem~\ref{thm:general-fixed-key-assisted} then implies the LOCC Fisher information scales as $\Omega\left((\log m)^{-2}\right)$.

\yk{States with private correlations have been prepared and distributed in quantum key distribution (QKD) experiments, with representative demonstrations using quantum-dot photon sources \cite{basso2021quantum}, optically linked trapped ions \cite{nadlinger2022experimental}, and satellite-distributed photons \cite{yin2020entanglement}. These implementations demonstrate that secure-key generation is possible despite realistic imperfections. Furthermore, noisy four-photon states exhibiting a separation between the achievable private-key rate and distillable entanglement have been prepared and characterized \cite{dobek2011experimental,dobek2013experimental}. Our theorem identifies an additional use of these private communication resources for sensing, with appropriate assistance when needed. Since QKD is a major focus of quantum network development, this connection allows quantum sensing to benefit from advances in preparing and distributing states with positive private information.

}



\yk{\textit{Generality of the state formulation---}
Here we describe given a general state $\rho^{AB}$, how to prepare corresponding sensing states in Theorem~\ref{thm:general-fixed-key-assisted} with $P_X(\sigma_0)$ equal to $P(\rho^{AB})$. Specifically, any shared state $\rho^{AB}$ yielding positive private information $P(\rho^{AB})=I(Z;B)-I(Z;E)>0$ under some POVM $\{M_z^A\}_z$ can serve as a resource for the assisted sensing protocol in Theorem~\ref{thm:general-fixed-key-assisted}. Such a state $\rho^{AB}$ 
converted through local processing into the form $\sigma^{XBC}$.



If every nonzero POVM element $M_z^A$ has rank one, we identify $A$ with $X$. This is the special case where $C$ is trivial and $\rho^{AB}$ can be directly used.  The Fisher-information lower bound then holds without an assisting quantum channel. 
In this case, we find $P(\rho^{AB})=I(A\rangle B)_\rho$, the privacy condition reduces to positive coherent information and certifies positive one-way distillable entanglement. This behavior can in fact be expected from the structure underlying private correlations. Recall that the private state has the form \cite{horodecki2005secure}, $\gamma^{ABA'B'}=U(\ket{\psi^+}\bra{\psi^+}^{AB}\otimes \sigma^{A'B'})U^\dagger$, where $\ket{\psi^+}=\sum_{i=0}^{d-1}\ket{ii}/\sqrt{d}$, $\sigma^{A'B'}$ is arbitrary, and $U=\sum_{i,j}\ket{ij}\bra{ij}^{AB}\otimes U_{ij}^{A'B'}$ is a unitary controlled on $AB$. Without the shield  $A'B'$, the state reduces to a maximally entangled key state on $AB$, so privacy is trivially guaranteed by entanglement. The shield $A'B'$ instead provides a structured way to incorporate compatible noise while preserving privacy without comparable distillable entanglement; indeed, the distillable entanglement can be made arbitrarily small while maintaining a fixed nonzero distillable key~\cite{horodecki2005secure}. 
Since our theorem derives its sensing guarantee from the privacy properties of the state, we expect an analogous distinction between privacy and entanglement to emerge in our setting. 

In general, for a POVM containing higher-rank elements, 
$\sigma^{XBC}$ can be prepared by applying the local isometry $V^{A\rightarrow XC}=\sum_x |x\rangle^X\otimes\sqrt{M_x}^{A\rightarrow C}$, with $\ket{x}$ chosen as the eigenbasis of the encoding Hamiltonian $H_X$:
\begin{equation}
\begin{aligned}
&\sigma_0^{XBC}
=
(V^{A\rightarrow XC}\otimes I_B)
\rho^{AB}
(V^{A\rightarrow XC}\otimes I_B)^\dagger.
\end{aligned}
\end{equation}
Measuring $X$ in the basis $\{|x\rangle\}$ and tracing out $C$ reproduces the outcome probabilities and conditional states of Bob $B$ and environment $E$ obtained by applying the original POVM to $\rho^{AB}$. Consequently, $P_X(\sigma_0)=P(\rho^{AB})>0$, satisfying the fixed-$X$ private-information condition of Theorem~\ref{thm:general-fixed-key-assisted}. This construction justifies the chosen state structure: the $XC$ decomposition arises from a local processing of the state $\rho^{AB}$ and need not be present in the original shared resource.
}





\textit{Relation to classical correlation} -- 
\yk{We construct an example showing that large classical correlation in the encoding basis does not, by itself, guarantee a large LOCC Fisher information.}
Intuitively, classical correlations depend only on the diagonal $X$ blocks, but the parameter encoding leaves these diagonal blocks unchanged
and encodes $\theta$ only on off-diagonal coherences between $b\neq b'$.
Thus, large classical correlation does not by itself guarantee a large
phase-sensitive coherence that determines the sensing capability.
The following example makes this separation explicit. 

Let $\dim X=\dim B=\dim C=n$, $b\in\{0,\ldots,n-1\}$. We consider the state
\begin{equation}
\begin{aligned}
\sigma_0^{XBC}=\beta\sum_b\ket{b,v_b}\bra{b,v_b}&+a\sum_{b\neq b'}\ket{b,v_{b'}}\bra{b,v_{b'}}\\
&+ta\sum_{b\neq b'}\ket{b,v_{b'}}\bra{b',v_b}.
\end{aligned}
\end{equation}
where $\ket{b,v_{b'}}=\ket{b}^{X}\ket{v_{b'}}^{BC}=\ket{b}^{X}\ket{b'}^{B}\ket{\chi_{b'}}^{C}$, $\ket{\chi_b}^{C}=\sum_{k=0}^{n-1}
\sqrt{p_k^{(n)}}\,e^{2\pi i bk/n}\ket{k}^{C}$, $\beta=\frac{u}{n}$, $a=\frac{1-u}{n(n-1)}$, $0<u,t<1$,     $p_k^{(n)}>0$, $\sum_{k=0}^{n-1}p_k^{(n)}=1$.
The parameter $\theta$ is encoded with
$U_\theta^X=e^{i\theta H_X}$, where
$H_X=\sum_bh_b\lvert b\rangle\langle b\rvert^X$,
$\operatorname{Tr}H_X=0$, and the coefficients $h_b$ may otherwise
be chosen arbitrarily. We introduce the first term as a phase-independent background which creates the large classical correlation; while the second and third terms provide the sensing capability of $\sigma$. We choose $\beta\geq a(1+t)$, which allows us to upper bound the locally accessible Fisher information. 


For fixed $u$,  the encoding-basis mutual information of $\Delta_X(\sigma_0)$ is 
\begin{equation}
    I(X;B)_\omega
    =
    u\log n-h_2(u)+o(1),\quad n\rightarrow\infty,
\end{equation}
where $h_2(u):=-u\log u-(1-u)\log (1-u)$ is the binary entropy, so the classical correlation can be arbitrarily large. But the
environment retains even more information about $X$, and we find
\begin{equation}
P_X(\sigma_0):=I(X;B)_\omega-I(X;E)_\omega<0,\quad n\geq3.
\end{equation}

Next we compare the global quantum Fisher information with the LOCC ones. We show that the quantum Fisher information is 
\begin{equation}
F_Q(\sigma_\theta)=\Theta\left(\frac{\operatorname{Tr}H_X^2}{n}\right).
\end{equation}
Let \(\mathcal{N}^{C\rightarrow C'}\) be any  parameter-independent channel acting on \(C\). 
We further show that both the unassisted and assisted LOCC Fisher information is upper bounded,
\begin{equation}
\begin{aligned}
&F_{\mathrm{LOCC}}^{(XC:B)}(\sigma_\theta)
=O\left(\frac{
\Delta_H^2R_n
}{
n
}\right)
,\quad R_n=2^{H_{1/2}(p^{(n)})},\\
&F_{\mathrm{LOCC}}^{(X:BC')}
\left[
\left(
\operatorname{id}^{XB}\otimes
\mathcal{N}^{C\to C'}
\right)(\sigma_\theta)
\right]
\leq
\frac{2}{n}F_Q(\sigma_\theta).
\end{aligned}    
\end{equation}
where $\Delta_H=\max_{b,b'}|h_b-h_{b'}|$, and $0\leq H_{1/2}(p^{(n)})\leq \log n$ denotes the Rényi entropy of order $1/2$. The factor $R_n$ can be kept bounded independently of $n$ by choosing a family $p^{(n)}$ whose $H_{1/2}$ remains bounded. For $h_b=\frac{b-(n-1)/2}{n-1}$ and a suitable family $p^{(n)}$ satisfying $R_n=O(1)$, we have $F_Q(\sigma_\theta)=\Theta(1)$, $F_{\mathrm{LOCC}}^{(X:BC')}=O(1/n)$, $F_{\mathrm{LOCC}}^{(XC:B)}=O(1/n)$.

The performance gap between LOCC and global measurement comes from at least two reasons: First, a global measurement in the Bell-like basis $\ket{\Psi_{bb'}^\pm}=(\ket{b,v_{b'}}\pm\ket{b',v_b})/\sqrt{2}$ simultaneously resolves the coherences $\ket{b,v_{b'}}\bra{b',v_b}$ for all $n(n-1)/2$ orthogonal pair subspaces, whereas LOCC pairwise readouts \(\ket{\phi_{bb',\pm}}^X\otimes\ket{\varphi_{bb',\pm}}^{BC}\), where \(\ket{\phi_{bb',\pm}}^X=(\ket{b}^X\pm\ket{b'}^X)/\sqrt{2}\) and \(\ket{\varphi_{bb',\pm}}^{BC}=(\ket{v_b}^{BC}\pm\ket{v_{b'}}^{BC})/\sqrt{2}\) rely on mutually incompatible local superpositions. Second, the background term $\beta\sum_b\ket{b,v_b}\bra{b,v_b}$ does not contribute to the global Bell-like outcomes but adds noise to LOCC product-measurement outcomes.

A key observation is that, although the mutual information in this example can grow logarithmically with the dimension, it does not yield a Fisher-information guarantee comparable to that obtained from private information in Theorem~\ref{thm:general-fixed-key-assisted}. When $\delta_H=\Theta(1/n)$ and $P_X=u\log n-O(1)$, the resulting lower bound on the LOCC Fisher information scales as $\Theta(n^{2u-2}(\log n)^{-2})$. By contrast, classical correlation provides no analogous guarantee: the sensing performance can remain $O(1/n)$ even when $I(X;B)_\omega=u\log_2 n-O(1)$. Therefore, when $u>1/2$, the LOCC Fisher information in this classically correlated example vanishes faster than the lower bound that Theorem~\ref{thm:general-fixed-key-assisted} would guarantee for the same amount of private information.

\textit{Channel interpretation} -- 
We can also formulate the problem in the channel picture, in which the probe is constructed from an input ensemble for a quantum channel and the sensing guarantee is expressed in terms of the channel’s private information.  Let $\mathcal{N}^{A'\to B}$ be a channel with Stinespring isometry $V_{\mathcal N}^{A'\to BE}$, and consider an input ensemble $\{p_x,\rho_x^{A'}\}_{x\in\mathcal X}$ with positive private information:
\begin{equation}
P^{(1)}(\mathcal{N}^{A'\to B},\{p_x,\rho_x^{A'}\})
:=
I(X;B)-I(X;E)>0. 
\end{equation}
For each $x$, choose $|\phi_x\rangle^{CA'}$ as purifications of $\rho_x^{A'}$. 
Alice uses these purifications to construct the coherent probe state
\begin{equation}
|\Psi\rangle^{XCA'}=
\sum_x\sqrt{p_x}
|x\rangle^X|\phi_x\rangle^{CA'} .
\end{equation}
Subsystem $A'$ is then transmitted to Bob through the channel $\mc N^{A'\rightarrow B}$ to form the distributed state. Meanwhile, the parameter is encoded on $X$ via $U_\theta^{X}=\exp(i\theta H_X)$ as before. Although written as an operation on $X$, $U_\theta^{X}$ provides an effective description of any sensing process that encodes an $x$-dependent phase on the coherent components. In particular, the physical encoding may act jointly on all subsystems, 
while we consider situations where the overall parameter dependence can be captured by $U_\theta^{X}$. 

The remaining discussion follows similarly from the state picture discussed earlier.
Alice subsequently sends subsystem $C$ to Bob through the assisting channel $\mc E_\lambda^{C \to \widetilde C}$ and produces the final state:
\begin{equation}
\begin{aligned}
&\sigma^{XB\widetilde C}_\theta=(( \mc E^{C \to \widetilde C}_\lambda \circ U_\theta^X\circ \mc N^{A'\rightarrow B}) )(\ket{\Psi}\bra{\Psi}^{XCA'}).   
\end{aligned}
\end{equation}
Alice has access to $X$ and Bob has access to $B\widetilde C$. We have
\begin{equation}
\mc F_{\mathrm{LOCC}}^{(X:B\widetilde C)}(\sigma^{XB\widetilde C}_\theta)
\ge
(1-\lambda)
\frac{\delta_H^2}{4L_m^2}
\left(2^{P^{(1)}}-1\right)^2.    
\end{equation}
The auxiliary system $C$ is introduced to purify the input ensemble. Its role can be understood from the relation:
\begin{equation}
\|( \operatorname{id}^{C}\otimes \mc N^{A'\rightarrow B})(\ket{\phi_x}\bra{\phi_{x'}}^{CA'})\|_1
=F(\sigma_x^E,\sigma_{x'}^E),   
\end{equation}
where $\sigma_x^E=\mc N^c(\rho_x^{A'})$ and $\mc N^c$ is the
complementary channel of $\mc N$. Since $P^{(1)}>0$ implies that the environment cannot determine $X$
perfectly; equivalently, the states $\{\sigma_x^E\}_x$ cannot be
perfectly distinguished. The coherence and the sensing capability are preserved when the purifying subsystems $C$ are included. 
If the input ensemble consists of pure states, $C$ is unnecessary and the private information reduces to the coherent information $I_c(\bar{\rho},\mathcal{N})$ of the average input $\bar{\rho}=\sum_xp_x\rho_x^{A'}$. 
Positive private information then implies positive one-shot coherent information and hence positive quantum capacity. Thus, in this case, the sensing lower bound is certified by a positive coherent-information witness of quantum capacity. 


Our channel construction mirrors the superactivation mechanism of Smith and Yard~\cite{smith2008quantum}, in which a zero-quantum-capacity channel with positive private information acquires positive joint capacity when combined with a $50\%$ erasure channel. In their construction, the erasure channel transmits a purification of the ensemble witnessing the private information, precisely the role played here by the auxiliary system $C$. Our result can therefore be viewed as a metrological analogue of the same activation mechanism: transmitting $C$ through the assisting channel converts privacy-protected coherence into LOCC-accessible Fisher information.


\textit{Conclusion and discussion} -- 
We have shown that private correlations provide a quantitative certificate of
distributed sensing capability in the presence of an assisting channel. Our data-hiding example exhibits a large separation between assisted and unassisted LOCC sensing performance, showing that transmitting only an unencoded shield subsystem can greatly enhance sensitivity, even through a noisy channel. More
generally, positive private information in the encoding basis guarantees a
nonzero lower bound on the LOCC Fisher information after the unencoded
subsystem is transmitted. The underlying mechanism is that incomplete knowledge of the encoding variable
by the environment forces coherence between different encoding branches to
remain in the state. The assisting channel then converts this coherence into an
LOCC-accessible sensing signal. This mechanism also clarifies why privacy  provides the relevant guarantee.
This perspective suggests that the hierarchy of correlations and capacities
developed in quantum Shannon theory may provide a corresponding hierarchy of
sensing guarantees, allowing useful sensing resources in noisy quantum systems
to be identified through well-studied operational properties. \yk{Extending our sensing guarantees to multipartite settings with suitable privacy criteria and arrangements of assisting channels is also an important direction for future work.}



\textit{Acknowledgments} -- We thank Paweł Horodecki for helpful discussions and comments.   YW, PW, and GS are supported under NSERC-NSF alliance grant ALLRP-586858-2023 and NSERC Discovery grant RGPIN-2025-02094. YW and SZ acknowledges funding provided by Perimeter Institute for Theoretical Physics, a research institute supported in part by the Government of Canada through the Department of Innovation, Science and Economic Development Canada and by the Province of Ontario through the Ministry of Colleges and Universities.   YW also acknowledges funding from the NSERC-UKRI Alliance Grant No. ALLRP-597823-24.  We acknowledge the use of ChatGPT (OpenAI) to assist with derivations and improve the presentation of this work; any AI-generated content was independently verified by the authors.

\bibliography{arxiv_PRL}

\clearpage

\onecolumngrid
\renewcommand{\appendixname}{Section}
\appendix


\vspace*{1cm}
\begin{center}
    {\Large\bfseries Supplemental Material: Private Correlations Certify Sensing Capability}
\end{center}
\vspace{1cm}
\setcounter{page}{1}

\section{Preliminaries}
\subsection{Private states}\label{SI_sec:pbit}

In this subsection, we briefly review private states. By measuring the key systems of a private state, Alice and Bob can obtain a perfectly secure classical key even when an eavesdropper holds a purification of the state \cite{horodecki2005secure}. Private states therefore play a role in secret-key distillation analogous to that played by maximally entangled states in entanglement distillation.

More precisely, a state $\gamma^{ABA'B'}$ is called an $l$-bit private state if measuring the key systems $A$ and $B$ in the computational basis produces a perfectly correlated $l$-bit key that is completely independent of the eavesdropper's purifying system. Here, $A$ and $B$ each consist of $l$ qubits, or equivalently have dimension $K=2^l$. Alice holds $AA'$, Bob holds $BB'$, and $A'$ and $B'$ are referred to as the shield systems.
A state with key systems $A,B$ and shield systems $A',B'$ is a private state if and only if it has the form \cite{horodecki2005secure}
\begin{equation}
    \gamma^{ABA'B'}
    =
    \frac{1}{K}\sum_{i,j=0}^{K-1}
    |ii\rangle\langle jj|^{AB}\otimes U_i^{A'B'} \sigma^{A'B'} (U_j^{A'B'})^\dagger,
\end{equation}
where \(K=2^l\) is the key dimension, \(\sigma^{A'B'}\) is a quantum state on \(A'B'\), and \(\{U_i^{A'B'}\}_{i=0}^{K-1}\) are unitary operators on the shield system.

In this work, we focus on the two-dimensional key subspace spanned by \(|N0\rangle\) and \(|0N\rangle\). Accordingly, we consider private states of the form
\begin{equation}
    \gamma^{ABA'B'}
    =
    \frac{1}{2}\Bigl(
    |N0\rangle\langle N0|^{AB}\otimes U_0\sigma U_0^\dagger
    +
    |0N\rangle\langle 0N|^{AB}\otimes U_1\sigma U_1^\dagger
    +
    |N0\rangle\langle 0N|^{AB}\otimes U_0\sigma U_1^\dagger
    +
    |0N\rangle\langle N0|^{AB}\otimes U_1\sigma U_0^\dagger
    \Bigr),
\end{equation}
where \(\sigma=\sigma^{A'B'}\) is a quantum state on the shield system \(A'B'\), and \(U_0,U_1\) are unitary operators on \(A'B'\).

We encode the parameter \(\theta\) on the key system by a phase shift on Alice's mode, namely
\[
    V_\theta^{AB}|N0\rangle = e^{-iN\theta}|N0\rangle,
    \qquad
    V_\theta^{AB}|0N\rangle = |0N\rangle.
\]
The resulting family of states is
\begin{equation}
\begin{aligned}
        \gamma^{ABA'B'}(\theta)
    &=
    \frac{1}{2}\Bigl(
    |N0\rangle\langle N0|^{AB}\otimes U_0\sigma U_0^\dagger
    +
    |0N\rangle\langle 0N|^{AB}\otimes U_1\sigma U_1^\dagger \\
    & +
    e^{-iN\theta}|N0\rangle\langle 0N|^{AB}\otimes U_0\sigma U_1^\dagger
    +
    e^{iN\theta}|0N\rangle\langle N0|^{AB}\otimes U_1\sigma U_0^\dagger
    \Bigr).
\end{aligned}
\end{equation}
For later use, given a quantum channel \(\mc E^{A'\to \wt B}\), we write
\begin{equation}
    \wt\gamma^{AB\wt B B'}(\theta)
    :=
    (\mc E^{A'\to \wt B}\otimes \mathrm{id}^{ABB'})(\gamma^{ABA'B'}(\theta)).
\end{equation}

\color{black}
\subsection{Fisher information preliminaries}
\label{SI_subsec:Fisher_information}

Throughout this section, Fisher information is evaluated locally at a fixed
operating point~$\theta_0$.

\begin{definition}[Classical Fisher information]
Let $\rho(\theta)$ be a differentiable family of quantum states and let
$M=\{M_x\}_x$ be a POVM.  Set
\[
    p(x|\theta):=\Tr\bigl(\rho(\theta)M_x\bigr).
\]
The classical Fisher information generated by $M$ at~$\theta_0$ is
\cite{kay1993statistical}
\begin{equation}
    F_{\mathrm{cl}}(\rho(\theta_0);M)
    :=
    \sum_{x:\,p(x|\theta_0)>0}
    \frac{
        \left.\bigl(\partial_\theta p(x|\theta)\bigr)^2
        \right|_{\theta=\theta_0}
    }{p(x|\theta_0)}.
\end{equation}
Terms with $p(x|\theta_0)=0$ are understood by the usual continuous-extension
convention whenever the limit exists; at the regular points considered below
they cause no difficulty.
\end{definition}

\begin{definition}[Unrestricted quantum Fisher information]
The quantum Fisher information (QFI) of $\rho(\theta)$ at~$\theta_0$ is
\cite{BraunsteinCaves1994}
\begin{equation}
    \mc F_Q(\rho(\theta_0))
    :=
    \sup_M F_{\mathrm{cl}}(\rho(\theta_0);M),
\end{equation}
where the supremum is over all POVMs.  Equivalently,
\begin{equation}
    \mc F_Q(\rho(\theta_0))
    =
    \Tr\bigl(\rho(\theta_0)L_{\theta_0}^2\bigr),
\end{equation}
where the symmetric logarithmic derivative $L_{\theta_0}$ satisfies
\begin{equation}
    \left.\partial_\theta\rho(\theta)\right|_{\theta=\theta_0}
    =
    \frac12\bigl(
        L_{\theta_0}\rho(\theta_0)
        +
        \rho(\theta_0)L_{\theta_0}
    \bigr).
\end{equation}
\end{definition}

\begin{definition}[LOCC Fisher information]
For a differentiable family of bipartite states $\rho^{AB}(\theta)$, define
\begin{equation}
    \mc F_{\mathrm{LOCC}}^{(A:B)}(\rho(\theta_0))
    :=
    \sup_{M\in\mathrm{LOCC}(A:B)}
    F_{\mathrm{cl}}(\rho(\theta_0);M).
\end{equation}
Here the supremum is over POVMs implementable by local operations and classical
communication across the bipartition $A:B$.
\end{definition}

The optimizing POVM may be calibrated using the known operating
point~$\theta_0$, but it is held fixed when the derivative with respect to
$\theta$ is evaluated.  In particular,
\begin{equation}
    \mc F_{\mathrm{LOCC}}^{(A:B)}(\rho(\theta_0))
    \le
    \mc F_Q(\rho(\theta_0)).
\end{equation}


\subsection{Relation to prior work}

Previous studies have investigated whether bound entanglement can serve as a resource for quantum metrology. 
Czekaj et al. constructed a family of multipartite bound-entangled states whose global quantum Fisher information exhibits Heisenberg scaling \cite{czekaj2015quantum}. Tóth and Vértesi established the existence of positive-partial-transpose (PPT) bound-entangled states that outperform all separable probes for metrology \cite{toth2018quantum}. Pál \emph{et al.} subsequently constructed high-dimensional families of PPT states whose metrological gain over separable probes approaches that of a maximally entangled qubit pair \cite{pal2021bound}. Tóth \emph{et al.}  showed that states exhibiting no metrological advantage in the single-copy setting can become useful after appending uncorrelated ancillary systems or using multiple copies, and proved that all bipartite entangled pure states outperform separable states in metrology \cite{toth2020activating}.

There are two key distinctions between our work and the prior studies~\cite{czekaj2015quantum,toth2018quantum,toth2020activating,pal2021bound}. First, these studies investigate entanglement as a metrological resource, including PPT bound entanglement. By contrast, we treat private correlations as the relevant resource and show that positive private information provides a quantitative certificate of sensing capability. Second, they characterize metrological usefulness by comparing the global quantum Fisher information with the maximum achievable by separable probes, whereas we derive a lower bound on the Fisher information accessible through LOCC in a distributed sensing setting. Conceptually, these earlier works ask whether and to what extent entanglement, including weak forms, can enhance metrology. Our results are instead rooted in the observation that the environment's incomplete knowledge imposes a lower bound on the total off-diagonal coherence. This privacy--coherence connection provides a distinct motivation for our approach and naturally suggests using privacy as an operational certificate of sensing capability.

The role of assistance in our work also differs fundamentally from the activation studied in Ref.~\cite{toth2020activating}. Their activation enlarges the probe by adding ancillary systems or additional copies and permits optimization over new encoding Hamiltonians acting on the enlarged local systems. In our setting, no new ancillary system or additional copy is appended, and the encoding Hamiltonian remains fixed. The assisting channel transfers only an already-present unencoded subsystem that never carries the unknown parameter. As a parameter-independent operation, this transmission cannot increase the global quantum Fisher information. Instead, it changes which party holds the unencoded subsystem and makes the surviving coherence accessible under the LOCC constraint. Our data-hiding example makes this distinction explicit by exhibiting a large separation between the locally accessible Fisher information before and after the transmission of the unencoded shield, even when the assisting channel is noisy.

We also note that one family of PPT states considered in Ref.~\cite{pal2021bound} is constructed from private bits. In that work, however, the private-state structure is used as a construction tool rather than as an operational resource criterion: neither private information nor distillable key is invoked to derive or bound the metrological performance. Consequently, no general bound connecting privacy to sensing capability is established. A separate distinction concerns the readout: the metrological performance of this family is evaluated by directly calculating its global quantum Fisher information, without restricting the readout to LOCC.

\cite{zhou2020saturating} showed that LOCC can attain the global quantum Fisher information for arbitrary multipartite pure states, and a mixed-state no-go example was also provided. In contrast to this prior qualitative result, our work is quantitative: rather than asking only whether the global quantum Fisher information is attainable,
we derive explicit lower and upper bounds on the Fisher information
accessible under LOCC.

\yk{
Coherence is conventionally regarded as a resource for quantum sensing \cite{streltsov2017colloquium}. Our work can be viewed as an effort to understand when coherence useful for sensing survives noise. Specifically, our results show that if private correlations survive, then useful coherence must also remain and can be made locally accessible with an appropriate assisting channel. The assisting channel thus provides a new mechanism for recovering access to the surviving coherence for sensing.

Quantum discord quantifies the minimum loss of mutual information when one subsystem is measured  and characterizes a general form of quantum correlation. A related discord-type measure, known as interferometric power, characterizes the worst-case quantum Fisher information over local generators whose spectrum is fixed but whose eigenbasis is unspecified \cite{girolami2014quantum}. Our work considers a fixed encoding and focuses on the Fisher information accessible through assisted LOCC measurements, addressing a different sensing setting.

Asymmetry characterizes the failure of a state to remain invariant under the parameter-encoding transformations. For the unitary encoding considered here, this corresponds to coherence between distinct eigenspaces of $H_X$, and quantum Fisher information is itself a measure of this asymmetry \cite{kudo2023fisher}. Asymmetry therefore provides a natural characterization of the resource underlying phase sensitivity. However, the presence of asymmetry alone does not specify how much sensing information is accessible under LOCC restrictions.

More importantly, compared with asymmetry measures and quantum discord, private correlations have a well established operational role in quantum communication, where their generation and preservation have been extensively investigated both theoretically and experimentally \cite{basso2021quantum,nadlinger2022experimental,yin2020entanglement,dobek2011experimental,dobek2013experimental}. Our work therefore provides a new way to identify a large class of quantum resources useful for sensing. Moreover, private correlations can remain substantial even when distillable entanglement is strongly suppressed, potentially allowing less restrictive requirements on the preparation and preservation of the distributed resource.

}

\section{Large gains of Fisher information: general hiding--recovery principle}\label{SI_sec:data_hiding_example}

\subsection{A general hiding--recovery principle}
\label{SI_subsec:hiding_recovery}
In this section, we provide further details of the data-hiding examples,
which demonstrate a gap between the sensing performance achievable by LOCC measurements
with and without an assisting channel.

Let $N\in\mbb N$ and define
\begin{equation}
    \ket{\psi_{N\pm}(\theta)}
    :=
    \frac{1}{\sqrt2}
    \bigl(\ket{N0}\pm e^{i N\theta}\ket{0N}\bigr),
    \qquad
    \psi_{N\pm}(\theta)
    :=
    \ket{\psi_{N\pm}(\theta)}
    \bra{\psi_{N\pm}(\theta)}.
\end{equation}
Given two states $\rho_+^{A'B'}$ and $\rho_-^{A'B'}$, consider
\begin{equation}\label{eq:general-hidden-sensing-state-app}
    \gamma_N^{ABA'B'}(\theta)
    :=
    \frac12\,\psi_{N+}^{AB}(\theta)\otimes\rho_+^{A'B'}
    +
    \frac12\,\psi_{N-}^{AB}(\theta)\otimes\rho_-^{A'B'}.
\end{equation}
The systems $AA'$ are initially held by Alice and $BB'$ by Bob.

Let $\mc E^{A'\to\wt B}$ be a quantum channel and set
\begin{equation}
    \sigma_\pm^{\wt B B'}
    :=
    \bigl(\mc E^{A'\to\wt B}\otimes\mathrm{id}^{B'}\bigr)
    (\rho_\pm^{A'B'}).
\end{equation}
After Alice transmits $A'$ through $\mc E$, the state becomes
\begin{equation}\label{eq:general-noisy-hidden-sensing-state-app}
    \wt\gamma_N^{AB\wt B B'}(\theta)
    :=
    \frac12\,\psi_{N+}^{AB}(\theta)\otimes\sigma_+^{\wt B B'}
    +
    \frac12\,\psi_{N-}^{AB}(\theta)\otimes\sigma_-^{\wt B B'},
\end{equation}
where Alice holds $A$ and Bob holds $B\wt B B'$.

We first present the elementary calculation for certificates of product POVM.

\begin{lemma}
\label{lem:local-two-mode-readout}
For $v\in[-1,1]$, let
\begin{equation}
\begin{aligned}
    \tau_{N,v}^{AB}(\theta)
    :=\frac12\Bigl(&
        \ket{N0}\bra{N0}
        +
        \ket{0N}\bra{0N} +
        v e^{-i N\theta}\ket{N0}\bra{0N}
        +
        v e^{i N\theta}\ket{0N}\bra{N0}
    \Bigr).
\end{aligned}
\end{equation}
At every operating point~$\theta_0$, there is a product POVM across $A:B$
whose classical Fisher information is $N^2v^2$.
\end{lemma}

\begin{proof}
For $s,t\in\{+1,-1\}$, define
\begin{equation}
    \ket{s_\alpha}
    :=
    \frac{1}{\sqrt2}\bigl(\ket0+s e^{i\alpha}\ket N\bigr),
    \qquad
    \ket{t_\beta}
    :=
    \frac{1}{\sqrt2}\bigl(\ket0+t e^{i\beta}\ket N\bigr).
\end{equation}
The corresponding product measurement gives
\begin{equation}
    p(s,t|\theta)
    =
    \frac14\Bigl[1+st\,v\cos(N\theta+\alpha-\beta)\Bigr].
\end{equation}
The measurement can be completed arbitrarily on the orthogonal complement of
$\operatorname{span}\{\ket0,\ket N\}$.
Choose $\alpha-\beta=\frac\pi2-N\theta_0$ and hold $\alpha,\beta$ fixed.
At $\theta=\theta_0$, all four probabilities equal $1/4$ and
\begin{equation}
    \left.\partial_\theta p(s,t|\theta)\right|_{\theta=\theta_0}
    =
    -\frac14st\,vN.
\end{equation}
The resulting classical Fisher information is therefore $N^2v^2$.
\end{proof}

We use the distinguishability norm induced by PPT measurements.  Denote
by $\mathrm{PPT}(A':B')$ the set of POVMs $\{L_z\}_z$ satisfying
\begin{equation}
    L_z\ge0,
    \qquad
    L_z^{T_{A'}}\ge0,
    \qquad
    \sum_zL_z=I^{A'B'}.
\end{equation}
For a Hermitian operator $X$ on $A'B'$, define
\begin{equation}\label{eq:PPT-distinguishability-norm}
    \|X\|_{\mathrm{PPT}(A':B')}
    :=
    \sup_{\{L_z\}_z\in\mathrm{PPT}(A':B')}
    \sum_z\bigl|\Tr(L_zX)\bigr|.
\end{equation}

\begin{thm}[General hiding--recovery principle]
\label{thm:hiding-recovery}
Let $\gamma_N(\theta)$ and $\wt\gamma_N(\theta)$ be defined by
\eqref{eq:general-hidden-sensing-state-app} and
\eqref{eq:general-noisy-hidden-sensing-state-app}.  Then, for every operating
point~$\theta_0$, the following statements hold.
\begin{enumerate}
    \item[\rm(i)] Before using the channel $\mc E$,
    \begin{equation}\label{eq:general-upper}
        \mc F_{\mathrm{LOCC}}^{(AA':BB')}
        (\gamma_N(\theta_0))
        \le
        N^2\|\rho_+-\rho_-\|_{\mathrm{PPT}(A':B')}.
    \end{equation}

    \item[\rm(ii)] Let $\{M_y\}_y$ be any POVM on $\wt B B'$ and define
    \begin{equation}
        a_y:=\Tr(M_y\sigma_+),
        \qquad
        b_y:=\Tr(M_y\sigma_-).
    \end{equation}
    After using the channel $\mc E$,
    \begin{equation}\label{eq:general-locc-lb}
        \mc F_{\mathrm{LOCC}}^{(A:B\wt B B')}
        (\wt\gamma_N(\theta_0))
        \ge
        N^2
        \sum_{y:\,a_y+b_y>0}
        \frac{(a_y-b_y)^2}{2(a_y+b_y)}.
    \end{equation}
\end{enumerate}
\end{thm}

The recovery protocol in part~\rm(ii) can in fact be implemented by a product
POVM across $A:B\wt B B'$.  Thus the lower bound does not require an
interactive LOCC protocol once the system $A'$ has been transmitted.

\begin{corollary}[Trace-norm recovery bound]
\label{cor:trace-distance}
For every operating point~$\theta_0$,
\begin{equation}\label{eq:trace-distance-lb}
    \mc F_{\mathrm{LOCC}}^{(A:B\wt B B')}
    (\wt\gamma_N(\theta_0))
    \ge
    \frac{N^2}{4}\|\sigma_+-\sigma_-\|_1^2.
\end{equation}
\end{corollary}

Theorem~\ref{thm:hiding-recovery} isolates the two ingredients of the
phenomenon.  Before transmission, the LOCC Fisher information is controlled by
the PPT distinguishability of the shield pair $(\rho_+,\rho_-)$.  After
transmission, the same phase coherence can be read out locally to the extent
that $(\sigma_+,\sigma_-)$ can be distinguished on Bob's side.  The parameter
$N$ fixes the absolute $N^2$ sensing scale, whereas the shield dimension will
control the relative enhancement.

\subsection{Concrete realization using Werner hiding states}
\label{SI_subsec:Werner_hiding}

Let $d\ge2$, let $S$ denote the swap operator on
$\mbb C^d\otimes\mbb C^d$, and set
\begin{equation}
    \Pi_+:=\frac{I+S}{2},
    \qquad
    \Pi_-:=\frac{I-S}{2}.
\end{equation}
The normalized symmetric and antisymmetric Werner states are
\begin{equation}\label{eq:Werner-hiding-pair}
    \rho_+
    :=
    \frac{\Pi_+}{\Tr(\Pi_+)}
    =
    \frac{I+S}{d(d+1)},
    \qquad
    \rho_-
    :=
    \frac{\Pi_-}{\Tr(\Pi_-)}
    =
    \frac{I-S}{d(d-1)}.
\end{equation}

\begin{lemma}[PPT distinguishability of the Werner pair]
\label{lem:Werner-PPT-norm}
For the states in~\eqref{eq:Werner-hiding-pair},
\begin{equation}
    \|\rho_+-\rho_-\|_{\mathrm{PPT}}
    =
    \frac{4}{d+1}.
\end{equation}
\end{lemma}

\begin{proof}
Set $\delta:=\rho_+-\rho_-$.  Since $\Tr(\delta)=0$, any PPT POVM can be
coarse-grained according to the sign of $\Tr(L_z\delta)$, and hence
\begin{equation}
    \|\delta\|_{\mathrm{PPT}}
    =
    2\sup_E\bigl|\Tr(E\delta)\bigr|,
\end{equation}
where both $E$ and $I-E$ are PPT.  Twirling $E$ by $U\otimes U$ preserves
these constraints and its expectation against $\delta$, so we may write
\begin{equation}
    E=x\Pi_+ + y\Pi_-,
    \qquad
    0\le x,y\le1.
\end{equation}
Replacing $E$ by $I-E$ if necessary, assume $x\ge y$, and write
\begin{equation}
    c:=\frac{x+y}{2},
    \qquad
    h:=\frac{x-y}{2}\ge0.
\end{equation}
Let $\phi_d:=\ket{\phi_d}\bra{\phi_d}$, where
$\ket{\phi_d}=d^{-1/2}\sum_{i=1}^d\ket{ii}$.  Since
$S^{T_{A'}}=d\phi_d$, the PPT constraints give
\begin{equation}
    0\le c+dh\le1.
\end{equation}
Since $y=c-h\ge0$, we have $c\ge h$ and therefore
$(d+1)h\le1$.  Consequently,
\begin{equation}
    \bigl|\Tr(E\delta)\bigr|
    =|x-y|
    =2h
    \le
    \frac{2}{d+1}.
\end{equation}
Equality is attained by $E=\frac{2}{d+1}\Pi_+$, which proves the claim.
\end{proof}

\begin{corollary}[Werner hiding bound]
\label{cor:werner-upper}
For the Werner pair in~\eqref{eq:Werner-hiding-pair},
\begin{equation}\label{eq:Werner-Fisher-upper}
    \mc F_{\mathrm{LOCC}}^{(AA':BB')}
    (\gamma_N(\theta_0))
    \le
    \frac{4N^2}{d+1}.
\end{equation}
\end{corollary}

We next show that simple noisy channels acting on Alice's shield recover
Fisher information of order $N^2$.

\begin{proposition}[Recovery through erasure and depolarizing channels]
\label{prop:erasure-depolarizing-recovery}
For the Werner pair in~\eqref{eq:Werner-hiding-pair}, the following statements
hold for every operating point~$\theta_0$.
\begin{enumerate}
    \item[\rm(i)] Let $\mc E_\lambda$ be the erasure channel
    \begin{equation}
        \mc E_\lambda(X)
        :=
        (1-\lambda)X
        +
        \lambda\Tr(X)\ket e\bra e,
        \qquad 0\le\lambda\le1,
    \end{equation}
    where $\ket e$ is orthogonal to the input space.  Then
    \begin{equation}\label{eq:erasure-recovery}
        \mc F_{\mathrm{LOCC}}^{(A:B\wt B B')}
        (\wt\gamma_N(\theta_0))
        \ge
        (1-\lambda)N^2.
    \end{equation}

    \item[\rm(ii)] Let $\mc D_p$ be the depolarizing channel
    \begin{equation}
        \mc D_p(X)
        :=
        (1-p)X+\frac{p}{d}\Tr(X)I,
        \qquad 0\le p\le1.
    \end{equation}
    Then
    \begin{equation}\label{eq:depolarizing-recovery}
        \mc F_{\mathrm{LOCC}}^{(A:B\wt B B')}
        (\wt\gamma_N(\theta_0))
        \ge
        \frac{(1-p)^2}{1-p^2/d^2}\,N^2.
    \end{equation}
\end{enumerate}
\end{proposition}

\begin{proof}
For the erasure channel, Bob uses the symmetric, antisymmetric, and erasure
outcomes.  Since both Werner states have marginal $I/d$, the corresponding
probabilities are
\begin{equation}
    a=(1-\lambda,0,\lambda),
    \qquad
    b=(0,1-\lambda,\lambda).
\end{equation}
Substitution into~\eqref{eq:general-locc-lb} gives
\eqref{eq:erasure-recovery}.

For the depolarizing channel,
\begin{equation}
    \sigma_\pm
    =
    (1-p)\rho_\pm+\frac{p}{d^2}I^{\wt B B'}.
\end{equation}
For the measurement $\{\Pi_+,\Pi_-\}$, the outcome probabilities satisfy
\begin{equation}
\begin{aligned}
    a_+-b_+&=1-p,
    &\qquad
    a_++b_+&=1+\frac pd,\\
    a_--b_-&=-(1-p),
    &
    a_-+b_-&=1-\frac pd.
\end{aligned}
\end{equation}
Substituting these expressions into~\eqref{eq:general-locc-lb} gives
\eqref{eq:depolarizing-recovery}.
\end{proof}

Combining Corollary~\ref{cor:werner-upper} with
Proposition~\ref{prop:erasure-depolarizing-recovery}, we obtain
\begin{equation}
    \frac{
        \mc F_{\mathrm{LOCC}}^{(A:B\wt B B')}
        (\wt\gamma_N(\theta_0))
    }{
        \mc F_{\mathrm{LOCC}}^{(AA':BB')}
        (\gamma_N(\theta_0))
    }
    \ge
    \frac{(1-\lambda)(d+1)}{4}
\end{equation}
for the erasure channel, and
\begin{equation}
    \frac{
        \mc F_{\mathrm{LOCC}}^{(A:B\wt B B')}
        (\wt\gamma_N(\theta_0))
    }{
        \mc F_{\mathrm{LOCC}}^{(AA':BB')}
        (\gamma_N(\theta_0))
    }
    \ge
    \frac{(d+1)(1-p)^2}{4(1-p^2/d^2)}
\end{equation}
for the depolarizing channel.  Thus, for every fixed $\lambda<1$ or $p<1$,
the relative enhancement grows linearly in $d$, while the recovered Fisher
information remains of order $N^2$.

\subsection{Proof of Theorem~\ref{thm:hiding-recovery}}
\label{SI_subsec:proof_hiding_recovery}

\begin{proof}[Proof of Theorem~\ref{thm:hiding-recovery}]
We prove the two statements separately.

\smallskip
\noindent\emph{Proof of~\rm(i).}
Set
\begin{equation}
    \omega:=\frac{\rho_++\rho_-}{2},
    \qquad
    \Delta:=\frac{\rho_+-\rho_-}{2}.
\end{equation}
Then
\begin{equation}
\begin{aligned}
    \gamma_N(\theta)
    ={}&
    \frac12
    \bigl(\ket{N0}\bra{N0}+\ket{0N}\bra{0N}\bigr)
    \otimes\omega \\
    &+
    \frac12
    \bigl(
        e^{-i N\theta}\ket{N0}\bra{0N}
        +
        e^{i N\theta}\ket{0N}\bra{N0}
    \bigr)
    \otimes\Delta.
\end{aligned}
\end{equation}

Let $\{M_y\}_y$ be an LOCC POVM across $AA':BB'$, and define the operators on
$A'B'$
\begin{equation}
    T_y:=\langle N0|M_y|0N\rangle,
    \qquad
    R_y:=\frac12\bigl(
        \langle N0|M_y|N0\rangle
        +
        \langle0N|M_y|0N\rangle
    \bigr).
\end{equation}
Write
\begin{equation}
    D_y:=\Tr(\omega R_y),
    \qquad
    t_y:=\Tr(\Delta T_y).
\end{equation}
The outcome probability is
\begin{equation}
    P(y|\theta)
    =
    D_y+\Re\bigl(e^{i N\theta}t_y\bigr).
\end{equation}
Since this expression is nonnegative for every~$\theta$, we have
$D_y\ge|t_y|$.

Fix~$\theta_0$ and set $z_y:=e^{i N\theta_0}t_y$.  For every outcome of
positive probability,
\begin{equation}\label{eq:FI-t-bound-main-proof}
\begin{aligned}
    \frac{
        \left.\bigl(\partial_\theta P(y|\theta)\bigr)^2
        \right|_{\theta=\theta_0}
    }{P(y|\theta_0)}
    & =
    N^2\frac{(\Im z_y)^2}{D_y+\Re z_y} \\
    & \le
    N^2\frac{(\Im z_y)^2}{|t_y|+\Re z_y} \\
    & =
    N^2\bigl(|t_y|-\Re z_y\bigr).
\end{aligned}
\end{equation}
The zero-denominator case is understood by continuity.  Moreover,
\begin{equation}
    \sum_y t_y
    =
    \Tr\left(
        \Delta \langle N0|\sum_yM_y|0N\rangle
    \right)
    =0.
\end{equation}
Summing~\eqref{eq:FI-t-bound-main-proof} over~$y$ therefore gives
\begin{equation}\label{eq:FI-sum-t}
    F_{\mathrm{cl}}(\gamma_N(\theta_0);\{M_y\}_y)
    \le
    N^2\sum_y|t_y|.
\end{equation}

It remains to control the off-diagonal blocks using a PPT measurement on the
shield.  Choose $\phi_y$ such that $t_y=|t_y|e^{i\phi_y}$, with arbitrary
$\phi_y$ when $t_y=0$.  For $\beta\in[0,2\pi)$, define
\begin{equation}
    \ket{a_{\pm}^{(y,\beta)}}
    :=
    \frac{1}{\sqrt2}
    \bigl(
        \ket N
        \pm e^{-i(\phi_y-\beta)}\ket0
    \bigr)^A,
    \qquad
    \ket{b_\beta}
    :=
    \frac{1}{\sqrt2}
    \bigl(\ket0+e^{-i\beta}\ket N\bigr)^B,
\end{equation}
and
\begin{equation}
    H_{y,\beta}^{\pm}
    :=
    \langle a_{\pm}^{(y,\beta)}\otimes b_\beta|
    M_y
    |a_{\pm}^{(y,\beta)}\otimes b_\beta\rangle.
\end{equation}
Each $M_y$ is separable across $AA':BB'$, so every
$H_{y,\beta}^{\pm}$ is PPT on $A':B'$.  Furthermore,
\begin{equation}
    \frac12\sum_{\pm}
    \ket{a_{\pm}^{(y,\beta)}}
    \bra{a_{\pm}^{(y,\beta)}}
    =
    \frac12\bigl(\ket0\bra0+\ket N\bra N\bigr),
\end{equation}
which is independent of~$y$.  Using $\sum_yM_y=I$, we obtain
\begin{equation}
    \sum_{y,\pm}\frac12H_{y,\beta}^{\pm}
    =I^{A'B'}.
\end{equation}
Hence $\{\frac12H_{y,\beta}^{\pm}\}_{y,\pm}$ is a PPT POVM for every~$\beta$.

A direct Fourier average gives
\begin{equation}\label{eq:key-identity-main-proof}
    \int_0^{2\pi}
    \bigl(H_{y,\beta}^+-H_{y,\beta}^-\bigr)
    \frac{d\beta}{2\pi}
    =
    \frac12\bigl(
        e^{-i\phi_y}T_y
        +
        e^{i\phi_y}T_y^\dagger
    \bigr).
\end{equation}
Indeed, the desired off-diagonal block is the zero Fourier mode, while all
other matrix elements carry factors $e^{\pm i \beta}$ or
$e^{\pm2i\beta}$.  Taking the trace against~$\Delta$ in
\eqref{eq:key-identity-main-proof} yields
\begin{equation}
    |t_y|
    =
    \int_0^{2\pi}
    \Tr\Bigl(
        \Delta(H_{y,\beta}^+-H_{y,\beta}^-)
    \Bigr)
    \frac{d\beta}{2\pi}.
\end{equation}
Consequently,
\begin{equation}
\begin{aligned}
    \sum_y|t_y|
    & \le
    \int_0^{2\pi}
    \sum_{y,\pm}
    \bigl|\Tr(\Delta H_{y,\beta}^{\pm})\bigr|
    \frac{d\beta}{2\pi} \\
    & \le
    2\|\Delta\|_{\mathrm{PPT}}
    =
    \|\rho_+-\rho_-\|_{\mathrm{PPT}}.
\end{aligned}
\end{equation}
Combining this estimate with~\eqref{eq:FI-sum-t} and taking the supremum over
LOCC POVMs proves~\eqref{eq:general-upper}.

\smallskip
\noindent\emph{Proof of~\rm(ii).}
Let $\{M_y\}_y$ be a POVM on $\wt B B'$.  Its outcome probability is
\begin{equation}
    p_y=\frac{a_y+b_y}{2},
\end{equation}
which is independent of~$\theta$.  Whenever $a_y+b_y>0$, the conditional state
on $AB$ is
\begin{equation}
    \tau_y^{AB}(\theta)
    =
    \tau_{N,v_y}^{AB}(\theta),
    \qquad
    v_y:=\frac{a_y-b_y}{a_y+b_y}.
\end{equation}
The phase choice in Lemma~\ref{lem:local-two-mode-readout} is independent of
$v_y$, so the same local product measurement can be used for every~$y$.
Because $p_y$ is independent of~$\theta$, the Fisher information of the joint
outcomes is
\begin{equation}
\begin{aligned}
    \sum_{y:\,a_y+b_y>0}p_yN^2v_y^2
    & =
    N^2\sum_{y:\,a_y+b_y>0}
    \frac{(a_y-b_y)^2}{2(a_y+b_y)}.
\end{aligned}
\end{equation}
This proves~\eqref{eq:general-locc-lb}.
\end{proof}

\begin{proof}[Proof of Corollary~\ref{cor:trace-distance}]
For any POVM $\{M_y\}_y$ on $\wt B B'$, Cauchy--Schwarz gives
\begin{equation}
\begin{aligned}
    \left(\sum_y|a_y-b_y|\right)^2
    & \le
    \left(\sum_y(a_y+b_y)\right)
    \left(\sum_{y:\,a_y+b_y>0}
        \frac{(a_y-b_y)^2}{a_y+b_y}
    \right) \\
    & =
    2\sum_{y:\,a_y+b_y>0}
        \frac{(a_y-b_y)^2}{a_y+b_y}.
\end{aligned}
\end{equation}
Choose a Helstrom measurement, for which
\begin{equation}
    \sum_y|a_y-b_y|=\|\sigma_+-\sigma_-\|_1,
\end{equation}
and apply~\eqref{eq:general-locc-lb}.
\end{proof}

\section{General states with a fixed key measurement}
\label{app:general-fixed-key}

In this appendix, we prove Theorem~\ref{thm:general-fixed-key-assisted}. We first show that positive private
information associated with the eigenbasis of the encoding Hamiltonian
forces coherence between different encoding branches. We then show
that these coherences can be read out by a one-way LOCC measurement
after Alice's unencoded system is transmitted to Bob. Combining these
two observations gives a lower bound on the locally accessible Fisher
information. After proving Theorem~\ref{thm:general-fixed-key-assisted}, we also
discuss the corresponding scenarios involving entanglement and
classical correlations.

Let $\sigma_0^{XBC}$ be a finite-dimensional state, and let
\begin{equation}
H_X=\sum_x h_x |x\rangle\langle x|_X.    
\end{equation}
The parameter-dependent state is
\begin{equation}
\sigma_\theta^{XBC}
=
(U_\theta^X\otimes I_{BC})
\sigma_0^{XBC}
(U_\theta^{X\dagger}\otimes I_{BC}),
\qquad
U_\theta^X=e^{i\theta H_X}.
\end{equation}
Writing $\sigma_0^{XBC}$ in the eigenbasis of $H_X$ gives
\begin{equation}
\sigma_0^{XBC}
=
\sum_{x,x'}
|x\rangle\langle x'|_X\otimes\sigma_{xx'}^{BC},
\end{equation}
and therefore
\begin{equation}
\sigma_\theta^{XBC}
=
\sum_{x,x'}
e^{i(h_x-h_{x'})\theta}
|x\rangle\langle x'|_X\otimes\sigma_{xx'}^{BC}.
\end{equation}
Define
\begin{equation}
p_x:=\operatorname{Tr}\sigma_{xx}^{BC},
\qquad
X_+:=\{x:p_x>0\},
\qquad
m:=|X_+|.
\end{equation}
For $x\in X_+$, define the normalized conditional states
\begin{equation}
\gamma_x^{BC}:=\frac{\sigma_{xx}^{BC}}{p_x},
\qquad
\gamma_x^B:=\operatorname{Tr}_C\gamma_x^{BC}.
\end{equation}
We also define
\begin{equation}
\delta_H:=\min_{\substack{x,x'\in X_+\\x\neq x'}}
|h_x-h_{x'}|.
\end{equation}
Relabel \(X_+=\{1,\ldots,m\}\) such that
\begin{equation}
h_1<h_2<\cdots<h_m,
\end{equation}
and define
\begin{equation}
L_m
:=
\max_{1\leq j\leq m}
\left(
\sum_{\ell=1}^{j-1}\frac{1}{\ell}
+
\sum_{\ell=1}^{m-j}\frac{1}{\ell}
\right).
\end{equation}
The nontrivial sensing statement requires \(m\geq2\) and \(\delta_H>0\).

\subsection{Branch coherence}

The lemma below follows by specializing known relations from the resource theory of coherence. Rana et al.~\cite{rana2017logarithmic} proved the corresponding relation for standard coherence, where the coherence is defined as the sum of the absolute values of all off-diagonal matrix elements, and anticipated its extension to arbitrary block decompositions by replacing this measure with the sum of the trace norms of the off-diagonal blocks. Bischof et al.~\cite{bischof2021quantifying} established the corresponding relations within a general framework for coherence with respect to quantum measurements. 
To apply these results in our setting, we specialize them to the encoding-basis decomposition and identify the resulting constraint on the total trace-norm coherence of the off-diagonal blocks relevant for sensing. We include a direct proof adapted to our notation to make this connection self-contained.

\begin{lemma}\label{lem:block-coherence-bound}
Let \(\rho^{XR} = \sum_{x,x'} |x\rangle \langle x'|^X\otimes A_{xx'}^R\) be a bipartite state, $p_x=\tr(A_{xx}^R)$, its purfication $\ket{\Phi}^{XRE}=\sum_x \sqrt{p_x}\ket{x}^X\ket{\phi_x}^{RE}$, $\omega^{XE}=\tr_R[\mathcal D_X(\ket{\Phi}\bra{\Phi}^{XRE})]$, and let \(\mathcal D_X\) denote complete dephasing on
\(X\) in the reference basis \(\{\ket{x}\}\). Then we have 
\begin{equation}
    H(X|E)_{\omega} \le\log (1+\sum_{x\neq x'} \|A_{xx'}\|_1),\quad \text{or}\quad
    \sum_{x\neq x'} \|A_{xx'}\|_1\geq 2^{H(X|E)_{\omega}}-1
\end{equation}
\end{lemma}

\begin{proof}
For each unordered pair \(x<x'\), let \(A_{xx'}=U_{xx'}|A_{xx'}|\) be the polar decomposition and define
\begin{equation}
P^{(xx')}
:=
\ket{x}\bra{x}^X\otimes |A_{xx'}^\dagger|
+\ket{x'}\bra{x'}^X\otimes |A_{xx'}|
-\ket{x}\bra{x'}^X\otimes A_{xx'}
-\ket{x'}\bra{x}^X\otimes A_{xx'}^\dagger .
\end{equation}
Since
\begin{equation}
P^{(xx')}
=
\begin{pmatrix}
U_{xx'}\\[1mm] -I
\end{pmatrix}
|A_{xx'}|
\begin{pmatrix}
U_{xx'}^\dagger & -I
\end{pmatrix}
\ge 0 ,
\end{equation}
we may set
\begin{equation}
\Omega:=\rho+\sum_{x<x'} P^{(xx')}\ge \rho .
\end{equation}
By construction, \(\Omega\) is block-diagonal in the \(X\)-basis. Moreover,
\begin{equation}
\Tr \Omega
=
1+\sum_{x<x'}\Bigl(\|A_{xx'}\|_1+\|A_{xx'}^\dagger\|_1\Bigr)
=
1+\sum_{x\neq x'} \|A_{xx'}\|_1.
\end{equation}
Hence, with
\begin{equation}
\sigma:=\frac{\Omega}{\Tr\Omega},
\end{equation}
we have \(\sigma\) block-diagonal in \(X\) and
\begin{equation}
\rho\le (1+\sum_{x\neq x'} \|A_{xx'}\|_1)\sigma .
\end{equation}
By operator monotonicity of the logarithm,
\begin{equation}
\log \rho
\le
\log(1+\sum_{x\neq x'} \|A_{xx'}\|_1) \cdot \mathbb I + \log \sigma .
\end{equation}
Therefore 
\begin{equation}
D(\rho\|\sigma)
=
\Tr \rho(\log \rho-\log \sigma)
\le
\log (1+\sum_{x\neq x'} \|A_{xx'}\|_1).
\end{equation}
And
\begin{equation}
D(\rho\|\sigma) =  S(\mathcal D_X(\rho))-S(\rho) + D(\mc D_X(\rho) \| \sigma)\geq S(\mathcal D_X(\rho))-S(\rho) .
\end{equation}
Note the relation
\begin{equation}\begin{aligned}
&X_{+}:=\left\{x: p_x>0\right\},\\
&S(\rho)=S(E)_\Phi=S(E)_\omega,\\
&S(\mathcal D_X(\rho))=H(X)+\sum_{x\in X_+}p_xS(A_{xx}/p_x)=H(X)+\sum_{x\in X_+}p_xS(E)_{\phi_x}=S(XE)_\omega  ,
\end{aligned}
\end{equation}
we have 
\begin{equation}
 S(\mathcal D_X(\rho))-S(\rho)=H(X|E)_\omega.
\end{equation}
Combining above equations, we have thus proven the lemma.
\end{proof}

\subsection{LOCC readout}\label{SI_sec:LOCC_readout}

In this subsection, we construct an LOCC measurement strategy that exploits the off-diagonal block coherence of $\sigma^{XBC}$. We first consider the simple case in which the measurement involves only one off-diagonal block.



\begin{lemma}\label{lem:two-block-locc}
Let
\begin{equation}
\rho_\theta^{XR}
=
\sum_{x,x'}
e^{i(h_x-h_{x'})\theta}
|x\rangle\langle x'|^X\otimes R_{xx'},\qquad
p_x:=\Tr R_{xx},
\end{equation}
where Alice holds $X$ and Bob holds $R$. For every $j\neq k$ with $p_j,p_k>0$, there exists a one-way LOCC measurement across $X:R$ such that
\begin{equation}
F_{\mathrm{LOCC}}^{(X:R)}\!\left(\rho_{\theta}^{XR}\right)
\ge
\frac{|h_j-h_k|^2\|R_{jk}\|_1^2}{p_j+p_k}.
\end{equation}
\end{lemma}

\begin{proof}
Set
\begin{equation}
\Delta_{jk}:=h_j-h_k .
\end{equation}
We only use the two-dimensional subspace of \(X\) spanned by
\(|j\rangle\) and \(|k\rangle\). Alice first projects onto this subspace by measuring in the basis
\begin{equation}
|\pm_\alpha\rangle
:=
\frac{|j\rangle \pm e^{i\alpha}|k\rangle}{\sqrt 2},
\end{equation}
where the phase \(\alpha\) will be chosen below.

The unnormalized conditional states on \(R\), given Alice's two outcomes, are
\begin{equation}
\tau_\pm(\theta)
=
\frac12
\left(
R_{jj}+R_{kk}
\pm e^{i(\Delta_{jk}\theta+\alpha)}R_{jk}
\pm e^{-i(\Delta_{jk}\theta+\alpha)}R_{kj}
\right).
\end{equation}
Therefore, at a given value $\theta_0$ of the unknown parameter,
\begin{equation}
\frac{ \partial \tau_\pm}{\partial \theta}\big|_{\theta=\theta_0}
=
\mp \Delta_{jk}\,
\operatorname{Im}
\left(
e^{i(\Delta_{jk}\theta_0+\alpha)}R_{jk}
\right)=\mp\Delta_{jk}H_\alpha.
\end{equation}
where $\operatorname{Im}(X)=(X-X^\dagger)/(2i)$, $\operatorname{Re}(X)=(X+X^\dagger)/2$, $H_\alpha$ is  Hermitian.

For any operator \(A\),
\begin{equation}
\|A\|_1
=
\|\operatorname{Re}A+i\operatorname{Im}A\|_1
\le
\|\operatorname{Re}A\|_1+\|\operatorname{Im}A\|_1 .
\end{equation}
Hence at least one of \(\operatorname{Re}A\) and \(\operatorname{Im}A\) has trace norm at least
\(\|A\|_1/2\). Now we choose \(\alpha\) so that $H_\alpha$ has largest trace norm, then
\begin{equation}
\|H_\alpha\|_1
\ge
\frac12\|R_{jk}\|_1 .
\end{equation}

Let
\begin{equation}
H_\alpha = H_{\alpha,+}-H_{\alpha,-}
\end{equation}
be the spectrum decomposition of \(H_\alpha\), where
\(H_{\alpha,+},H_{\alpha,-}\ge 0\) have orthogonal supports. Let
\(\Pi_+\) and \(\Pi_-\) be the projections onto the supports of
\(H_{\alpha,+}\) and \(H_{\alpha,-}\), respectively, and let
\(\Pi_0=I-\Pi_+-\Pi_-\).

Bob measures the POVM \(\{\Pi_+,\Pi_-,\Pi_0\}\). Denote the resulting
joint probabilities by
\begin{equation}
p_{s,b}(\theta):=\operatorname{Tr}[\Pi_b\tau_s(\theta)],
\qquad
s\in\{+,-\},\quad b\in\{+,-,0\}.
\end{equation}
For Alice outcome \(+\), we have
\begin{equation}
\begin{aligned}
&\dot p_{+,+}(\theta_0)
=
\operatorname{Tr}[\Pi_+\dot\tau_+(\theta_0)]
=
-\Delta_{jk}\operatorname{Tr}H_{\alpha,+},\\
&\dot p_{+,-}(\theta_0)
=
\operatorname{Tr}[\Pi_-\dot\tau_+(\theta_0)]
=
\Delta_{jk}\operatorname{Tr}H_{\alpha,-},\\  
&\dot p_{+,0}(\theta_0)=0
\end{aligned}
\end{equation}

Therefore
\begin{equation}
\sum_{b\in\{+,-,0\}}|\dot p_{+,b}(\theta_0)|
=
|\Delta_{jk}|
\left(
\operatorname{Tr}H_{\alpha,+}
+
\operatorname{Tr}H_{\alpha,-}
\right)
=
|\Delta_{jk}|\|H_\alpha\|_1 .
\end{equation}
Similarly, for Alice outcome \(-\), 
\begin{equation}
\begin{aligned}
&\dot p_{-,+}(\theta_0)
=
\Delta_{jk}\operatorname{Tr}H_{\alpha,+}, \\
&\dot p_{-,-}(\theta_0)
=
-\Delta_{jk}\operatorname{Tr}H_{\alpha,-},\\
&\dot p_{-,0}(\theta_0)=0.
\end{aligned}
\end{equation}
Hence
\begin{equation}
\sum_{b\in\{+,-,0\}}|\dot p_{-,b}(\theta_0)|
=
|\Delta_{jk}|\|H_\alpha\|_1 .
\end{equation}

Combining the two Alice outcomes gives
\begin{equation}
\sum_{s\in\{+,-\}}\sum_{b\in\{+,-,0\}}
|\dot p_{s,b}(\theta_0)|
=
2|\Delta_{jk}|\|H_\alpha\|_1 .
\end{equation}
By the choice of \(\alpha\),
\begin{equation}
\|H_\alpha\|_1\ge \frac12\|R_{jk}\|_1.
\end{equation}
Thus, we can find that
\begin{equation}
\sum_{s\in\{+,-\}}\sum_{b\in\{+,-,0\}}|\dot p_{s,b}(\theta_0)|
\ge
|\Delta_{jk}|\|R_{jk}\|_1 .
\end{equation}


The complementary outcome of Alice's projection has probability
$1-p_j-p_k$, which is independent of $\theta$, and therefore does not
contribute to the Fisher information. Moreover, the total probability
of the outcomes in the $j,k$ sector is
\begin{equation}
\sum_{s\in\{+,-\}}\sum_{b\in\{+,-,0\}}p_{s,b}(\theta_0)
=
p_j+p_k.
\end{equation}
Applying the weighted Cauchy--Schwarz inequality only to these outcomes,
we obtain
\begin{equation}
\begin{aligned}
F_{\mathrm{cl}}
&=
\sum_{s\in\{+,-\}}\sum_{b\in\{+,-,0\}}
\frac{\dot p_{s,b}(\theta_0)^2}{p_{s,b}(\theta_0)}\ge
\frac{
\left(
\sum_{s\in\{+,-\}}\sum_{b\in\{+,-,0\}}
|\dot p_{s,b}(\theta_0)|
\right)^2
}{
\sum_{s\in\{+,-\}}\sum_{b\in\{+,-,0\}}
p_{s,b}(\theta_0)
}\ge
\frac{|h_j-h_k|^2\|R_{jk}\|_1^2}{p_j+p_k}.
\end{aligned}
\end{equation}
\end{proof}

We now show that all off-diagonal $X$ blocks can be incorporated into
a single one-way LOCC measurement. Although the two-dimensional subspaces associated with different pairs generally overlap, Alice can combine the corresponding pairwise measurements into a valid POVM by assigning the pair $j,k$ a weight proportional to $1/(k-j)$ after ordering $h_1<\cdots<h_m$.

\begin{lemma}\label{lemm:all-pair-locc}
Let
\begin{equation}
\rho_\theta^{XR}
=
\sum_{j,k\in X_+}
e^{i(h_j-h_k)\theta}
|j\rangle\langle k|_X\otimes R_{jk}, \quad p_j=\operatorname{Tr}R_{jj}.
\end{equation}
where Alice holds $X$ and Bob holds $R$.
There exists a one-way LOCC
measurement from $X$ to $R$ such that
\begin{equation}
F_{\mathrm{LOCC}}^{(X: R)}
\left(\rho_{\theta}^{XR}\right)
\geq
\frac{1}{L_m}
\sum_{\substack{j<k\\j,k\in X_+}}
\frac{|h_j-h_k|^2}{k-j}
\frac{\|R_{jk}\|_1^2}{p_j+p_k}.
\end{equation}
\end{lemma}

\begin{proof}
For each unordered pair $j<k$ with $j,k\in X_+$, choose the phase
$\alpha_{jk}$ and Bob's conditional POVM according to the two-block
protocol constructed in the proof of Lemma~\ref{lem:two-block-locc} at $\theta_0$. Since Alice's pairwise measurements for different choices of
$j,k$ generally overlap, they must be appropriately weighted to form a
valid POVM. Specifically, Alice's
measurement is chosen as
\begin{equation}
\label{eq:all-pair-povm}
|\pm_{\alpha_{jk}}\rangle
:=
\frac{|j\rangle\pm e^{i\alpha_{jk}}|k\rangle}{\sqrt{2}},
\qquad
E_{jk,\pm}
:=
\frac{1}{L_m(k-j)}
|\pm_{\alpha_{jk}}\rangle
\langle\pm_{\alpha_{jk}}|.
\end{equation}
For each pair of $j,k$,
\begin{equation}
E_{jk,+}+E_{jk,-}
=
\frac{1}{L_m(k-j)}
\left(|j\rangle\langle j|+|k\rangle\langle k|\right).
\end{equation}
Therefore, the total coefficient of \(|j\rangle\langle j|\) in the sum of all $E_{jk,\pm}$ is
\begin{equation}
\frac{1}{L_m}
\sum_{\substack{k=1\\k\neq j}}^m
\frac{1}{|j-k|}
\leq 1.
\end{equation}
It follows that
\begin{equation}
E_\perp
:=
I_X-
\sum_{\substack{j<k\\j,k\in X_+}}
\sum_{s\in\{+,-\}}E_{jk,s}
\geq0,
\end{equation}
so these operators form a POVM. Moreover, \(E_\perp\) is diagonal in the encoding basis, and hence its outcome probability is independent of \(\theta\).

After obtaining \((j,k,s)\), Alice sends this label to Bob, who performs the corresponding conditional POVM from Lemma~\ref{lem:two-block-locc}. Let \(p_{s,b}^{(jk)}(\theta)\) denote the probabilities of the outcomes in the one-pair protocol in Lemma~\ref{lem:two-block-locc} for pair $j,k$. The proof of Lemma~\ref{lem:two-block-locc} establishes
\begin{equation}
\sum_{s,b}
\frac{\dot p_{s,b}^{(jk)}(\theta_0)^2}
{p_{s,b}^{(jk)}(\theta_0)}
\geq
\frac{|h_j-h_k|^2\|R_{jk}\|_1^2}{p_j+p_k}.
\end{equation}
For the weighted all-pair POVM in Eq.~\ref{eq:all-pair-povm}, the corresponding probabilities are
\begin{equation}\label{eq:P_jksb}
\widetilde p_{jk,s,b}(\theta)
=
\frac{1}{L_m(k-j)}
p_{s,b}^{(jk)}(\theta).
\end{equation}
Consequently,
\begin{equation}
\sum_{s,b}
\frac{\dot{\widetilde p}_{jk,s,b}(\theta_0)^2}
{\widetilde p_{jk,s,b}(\theta_0)}
\geq
\frac{1}{L_m(k-j)}
\frac{|h_j-h_k|^2\|R_{jk}\|_1^2}{p_j+p_k}.
\end{equation}
The outcomes associated with different pairs $j,k$ contribute additively to the Fisher information. Hence,
\begin{equation}
F_{\mathrm{LOCC}}^{(X: R)}
\left(\rho_{\theta}^{XR}\right)
\geq
\frac{1}{L_m}
\sum_{\substack{j<k\\j,k\in X_+}}
\frac{|h_j-h_k|^2}{k-j}
\frac{\|R_{jk}\|_1^2}{p_j+p_k},
\end{equation}
which proves the lemma.
\end{proof}

\subsection{Proof of Theorem 2}

We first consider the state obtained when Alice's unencoded system
$C$ is transmitted perfectly to Bob. 
Applying Lemma~\ref{lemm:all-pair-locc} to \(\sigma_{\theta}^{XBC}\), with \(R=BC\), gives
\begin{equation}
F_{\mathrm{LOCC}}^{(X: BC)}
\left(\sigma_{\theta}^{XBC}\right)
\geq
\frac{1}{L_m}
\sum_{j<k}
\frac{|h_j-h_k|^2}{k-j}
\frac{\|\sigma_{jk}^{BC}\|_1^2}{p_j+p_k}.
\end{equation}
Because the eigenvalues are ordered,
\begin{equation}
|h_j-h_k|
=
\sum_{\ell=j}^{k-1}(h_{\ell+1}-h_\ell)
\geq
\delta_H(k-j).
\end{equation}
Therefore,
\begin{equation}
F_{\mathrm{LOCC}}^{(X: BC)}
\left(\sigma_{\theta}^{XBC}\right)
\geq
\frac{\delta_H^2}{L_m}
\sum_{j<k}
(k-j)\frac{\|\sigma_{jk}^{BC}\|_1^2}{p_j+p_k}.
\end{equation}
Let
\begin{equation}
T=\sum_{j<k}\|\sigma_{jk}^{BC}\|_1.
\end{equation}
By weighted Cauchy--Schwarz inequality,
\begin{equation}
T^2
\leq
\left[
\sum_{j<k}
(k-j)\frac{\|\sigma_{jk}^{BC}\|_1^2}{p_j+p_k}
\right]
\left[
\sum_{j<k}\frac{p_j+p_k}{k-j}
\right].
\end{equation}
The second factor satisfies
\begin{equation}
\sum_{j<k}\frac{p_j+p_k}{k-j}
=
\sum_{j=1}^m
p_j
\left(
\sum_{\ell=1}^{j-1}\frac{1}{\ell}
+
\sum_{\ell=1}^{m-j}\frac{1}{\ell}
\right)
\leq L_m.
\end{equation}
It follows that
\begin{equation}
F_{\mathrm{LOCC}}^{(X: BC)}
\left(\sigma_{\theta}^{XBC}\right)
\geq
\frac{\delta_H^2}{L_m^2}T^2.
\end{equation}
By Lemma~\ref{lem:block-coherence-bound} and the Hermiticity of \(\sigma_0^{XBC}\),
\begin{equation}
T
=
\frac{1}{2}\sum_{j\neq k}\|\sigma_{jk}^{BC}\|_1
\geq
\frac{2^{H(X|E)_\omega}-1}{2}.
\end{equation}
Consequently,
\begin{equation}
F_{\mathrm{LOCC}}^{(X: BC)}
\left(\sigma_{\theta}^{XBC}\right)
\geq
\frac{\delta_H^2}{4L_m^2}
\left(2^{H(X|E)_\omega}-1\right)^2.
\end{equation}

We now consider an erasure channel with erasure probability
$\lambda=1-t$. The output state has two
orthogonal flag sectors:
\begin{equation}
\widetilde{\sigma}_\theta^{XB\widetilde C}
=
t\,\sigma_\theta^{XBC}
\oplus
(1-t)\,\sigma_\theta^{XB}.
\label{eq:erasure-direct-sum}
\end{equation}
Alice performs the all-pair POVM defined in
Eq.~\eqref{eq:all-pair-povm} and sends her  outcomes to
Bob. Bob first reads the erasure flag. In the transmitted sector, he
performs the corresponding conditional measurement on $BC$ as constructed in Lemma~\ref{lem:two-block-locc}. In the
erased sector, he maps all outcomes to a single failure symbol. This
protocol uses only one-way communication from Alice to Bob.

Let $q_y(\theta)$ denote the outcome probabilities of the
perfect-transmission protocol as in Eq.~\ref{eq:P_jksb}, where $y=(jk,s,b)$. The corresponding outcome
probabilities in the transmitted sector of the erasure channel
are
\begin{equation}
\widetilde q_y(\theta)
=
tq_y(\theta).
\end{equation}
Since $t$ is independent of $\theta$,
\begin{equation}
\dot{\widetilde q}_y(\theta)
=
t\dot q_y(\theta).
\end{equation}
The Fisher-information contribution of the transmitted sector is
therefore
\begin{equation}
\begin{aligned}
\sum_y
\frac{
\dot{\widetilde q}_y(\theta_0)^2
}{
\widetilde q_y(\theta_0)
}
&=
\sum_y
\frac{
t^2\dot q_y(\theta_0)^2
}{
tq_y(\theta_0)
}=
t
\sum_y
\frac{
\dot q_y(\theta_0)^2
}{
q_y(\theta_0)
}.
\end{aligned}
\end{equation}
The failure symbol has probability $1-t$, independent of $\theta$,
and hence contributes no Fisher information. It follows that
\begin{equation}
F_{\rm LOCC}^{(X: B\widetilde C)}
\left(
\widetilde{\sigma}_\theta^{XB\widetilde C}
\right)
\geq
t
F_{\rm LOCC}^{(X: BC)}
\left(
\sigma_{\theta}^{XBC}
\right)\geq
\frac{(1-\lambda)\delta_H^2}{4L_m^2}
\left(
2^{H(X|E)_\omega}-1
\right)^2.
\label{eq:erasure-postselection}
\end{equation}

To connect this metrological bound to an operational measure of privacy, we now relate it to the fixed-$X$ private information. Recall that
\begin{equation}
P_X(\sigma_0)
=
I(X;B)_\omega-I(X;E)_\omega
=
H(X|E)_\omega-H(X|B)_\omega.
\end{equation}
Because $\omega^{XB}$ is a cq state,
\begin{equation}
H(X|B)_\omega\ge0,
\end{equation}
and therefore
\begin{equation}
H(X|E)_\omega\ge P_X(\sigma_0).
\end{equation}
Under the assumption $P_X(\sigma_0)>0$, both
$2^{H(X|E)_\omega}-1$ and $2^{P_X(\sigma_0)}-1$ are
nonnegative. Hence
\begin{equation}
\left(2^{H(X|E)_\omega}-1\right)^2
\ge
\left(2^{P_X(\sigma_0)}-1\right)^2.
\end{equation}
Consequently,
\begin{equation}
F_{\mathrm{LOCC}}^{(X: B\widetilde C)}
\bigl(\widetilde{\sigma}_{\theta}^{XB\widetilde C}\bigr)
\ge
\frac{(1-\lambda)\delta_H^2}{4L_m^2}
\left(2^{P_X(\sigma_0)}-1\right)^2.
\end{equation}
This proves Theorem~\ref{thm:general-fixed-key-assisted}.

For example, for the $50\%$ erasure channel, $\lambda=1/2$, and thus
\begin{equation}
F_{\mathrm{LOCC}}^{(X: B\widetilde C)}
\bigl(\widetilde{\sigma}_{\theta}^{XB\widetilde C}\bigr)
\ge
\frac{\delta_H^2}{8L_m^2}
\left(2^{P_X(\sigma_0)}-1\right)^2.
\end{equation}

\subsection{Relation to the private-bit example}\label{SI_sec:relation_pbit}

In this subsection, we verify that the lower bound on the LOCC Fisher information
for the data-hiding example discussed in
Theorem~\ref{main_thm:hiding-recovery} follows as a special case of
Theorem~\ref{thm:general-fixed-key-assisted}. Consider
\begin{equation}
\gamma_N^{ABA'B'}(\theta)
:=
\frac12\,\psi_{N+}^{AB}(\theta)\otimes \rho_+^{A'B'}
+
\frac12\,\psi_{N-}^{AB}(\theta)\otimes \rho_-^{A'B'},
\end{equation}
where
\begin{equation}
\rho_+^{A'B'}=\frac{I+S}{d(d+1)},
\qquad
\rho_-^{A'B'}=\frac{I-S}{d(d-1)}.
\end{equation}
We apply Theorem~\ref{thm:general-fixed-key-assisted} by
taking \(X=A\), \(C=A'\), and regarding \(BB'\) as Bob's system \(B\).
For convenience, we identify
\(\ket{0}^{X}\equiv\ket{N}^{A}\) and
\(\ket{1}^{X}\equiv\ket{0}^{A}\). It follows that
\(p_0=p_1=1/2\), \(m=2\), \(\delta_H=N\), and \(L_2=1\).

Note that \(\gamma_N(0)\) is a private bit. To see this explicitly, recall that the symmetric and antisymmetric Werner states are supported on the \(+1\) and \(-1\) eigenspaces of the swap operator \(S\),  and therefore satisfy \(S\rho_\pm=\pm\rho_\pm\). Setting
\(\sigma=(\rho_++\rho_-)/2\), \(U_0=I\), and \(U_1=S\) in the
standard form of private states reviewed in the Sec.~\ref{SI_sec:pbit} of the Supplemental Material gives
\begin{align}
U_0\sigma U_0^\dagger
&=
U_1\sigma U_1^\dagger
=
\frac{\rho_++\rho_-}{2},\\
U_0\sigma U_1^\dagger
&=
U_1\sigma U_0^\dagger
=
\frac{\rho_+-\rho_-}{2}.
\end{align}
Substituting these relations into the standard form of private state yields
\begin{align}
\gamma^{ABA'B'}
={}&
\frac{1}{4}
\left(
\ket{N0}\bra{N0}^{AB}
+
\ket{0N}\bra{0N}^{AB}
\right)
\otimes
\left(\rho_++\rho_-\right)^{A'B'}
\nonumber\\
&+
\frac{1}{4}
\left(
\ket{N0}\bra{0N}^{AB}
+
\ket{0N}\bra{N0}^{AB}
\right)
\otimes
\left(\rho_+-\rho_-\right)^{A'B'}
\nonumber\\
={}&
\frac{1}{2}\psi_{N+}^{AB}(0)\otimes\rho_+^{A'B'}
+
\frac{1}{2}\psi_{N-}^{AB}(0)\otimes\rho_-^{A'B'}
=
\gamma_N(0).
\end{align}
Thus, \(\gamma_N(0)\) has exactly the private-state form with a one-bit key.
Although $\gamma_N(0)$ is a private bit, the standard definition of private states refers to the secret key obtained by measuring both key systems $A$ and $B$ in the computational basis. By contrast, Theorem~\ref{thm:general-fixed-key-assisted} requires positive private information $P_X$ between the measured system $X=A$ and Bob's systems $BB'$. We show explicitly below that $P_X(\gamma_N(0))=1$.

Under the identification of \(X\) with Alice's key system \(A\), measuring \(X\) produces the two outcomes with equal probabilities. The defining security property of a private state implies that the conditional state of any purifying environment \(E\) is the same for both outcomes. Consequently, the resulting \(XE\) state has the product form
\begin{equation}
\omega^{XE}
=
\frac{I^X}{2}\otimes\omega^E,
\end{equation}
and therefore \(I(X;E)_\omega=0\).

At the same time, Alice's two key states \(\ket{N}^{A}\) and \(\ket{0}^{A}\) are paired with the mutually orthogonal states \(\ket{0}^{B}\) and \(\ket{N}^{B}\), respectively. Bob can therefore determine \(X\) perfectly by measuring his key system \(B\). Hence,
\begin{equation}
H(X|BB')_\omega=0,
\end{equation}
and, since \(H(X)_\omega=1\),
\begin{equation}
I(X;BB')_\omega=H(X)_\omega-H(X|BB')_\omega=1.
\end{equation}
Therefore,
\(P_X(\gamma_N(0))=1\).

Substituting these values into
Theorem~\ref{thm:general-fixed-key-assisted} gives an assisted LOCC
Fisher-information lower bound of \((1-\lambda)N^2/4\). This bound is independent
of the shield dimension \(d\) and therefore reproduces the lower bound established in Theorem~\ref{main_thm:hiding-recovery}. However, the lower bound from Theorem~\ref{thm:general-fixed-key-assisted} is
not tight for this example, the direct recovery calculation in Theorem~\ref{main_thm:hiding-recovery} gives
\((1-\lambda)N^2\), which is larger by a factor of four because it
exploits the specific binary structure of the Werner hiding pair. 

We emphasize that Theorem~\ref{thm:general-fixed-key-assisted}
applies to any positive value of the private information, including
noninteger values, whereas the private-state-based data-hiding example
in Theorem~\ref{main_thm:hiding-recovery} has integer-valued private
information.
The integer-valued private information in the data-hiding example
follows from the standard form of private-state used here, whose key systems encode an integer number of private bits.
Note also that Theorem~\ref{thm:general-fixed-key-assisted} does not provide an upper bound on the unassisted LOCC Fisher
information, it only provides a lower bound on the assisted LOCC
Fisher information and thus guarantees a certain level of assisted
sensing performance.
Nevertheless, the data-hiding example in
Theorem~\ref{main_thm:hiding-recovery} demonstrates that this assisted
lower bound is nontrivial, since the LOCC Fisher information is
limited in the absence of the assisting channel in some cases.

We can also evaluate the logarithmic negativity of $\gamma_N(0)$,
\begin{equation}
E_N(\gamma_N(0))=\log\|\gamma_N(0)^{T_{BB'}}\|_1=\log(1+2/d),
\end{equation}
which provides an upper bound on its distillable entanglement \cite{vidal2002computable,plenio2005logarithmic}.
Therefore, the distillable entanglement is $O(1/d)$ as $d\to\infty$.

\yk{\subsection{From general states with positive private information to the $\sigma^{XBC}$ form used in Theorem~\ref{thm:general-fixed-key-assisted}}

In this subsection, we justify the state structure $\sigma^{XBC}$ and the requirement of positive private information in a fixed encoding basis of $X$. We show that any shared state $\rho^{AB}$ yielding positive private information under a POVM $\{M_z\}_z$ can serve as a sensing resource within our framework. Such a state can either be used directly with a trivial subsystem $C$, or be converted through local coherent processing into a state $\sigma^{XBC}$ of the form considered in Theorem \ref{thm:general-fixed-key-assisted}. The partition between $X$ and $C$ arises naturally from this coherent processing, as we explain in detail below.

We consider a general shared state $\rho^{AB}$, with purification $|\Psi\rangle^{ABE}$, and a POVM $\{M_z^A\}_z$ on Alice's system. Let $Z$ denote the classical measurement record, $q_z$ the probability of outcome $z$, and $\tau_z^{BE}$ the corresponding normalized conditional state of Bob and Eve. The measurement produces
\begin{equation}\label{SI_eq:omega_ZBE}
\omega^{ZBE}
=
\sum_z q_z |z\rangle\langle z|^Z\otimes\tau_z^{BE}.
\end{equation}
Suppose that this measurement yields positive private information,
\begin{equation}
P=I(Z;B)_\omega-I(Z;E)_\omega>0.
\end{equation}
This is a standard sufficient condition for private-key generation. We show how such a resource state $\rho^{AB}$ fits our sensing framework by distinguishing two cases: rank-one measurements, for which the unencoded system $C$ can be taken to be trivial, and measurements with higher-rank elements, for which coherent processing constructs the systems $X$ and $C$.

\paragraph{Rank-one measurements and a trivial unencoded system.}

Suppose every nonzero POVM element has rank one,
\begin{equation}
M_z^A=|\psi_z\rangle\langle\psi_z|^A,
\end{equation}
where the vectors $|\psi_z\rangle$ need not to be normalized or mutually orthogonal. For each outcome $z$, the conditional state of $BE$ is pure
\begin{equation}
\tau_z^{BE}
=
|\eta_z\rangle\langle\eta_z|^{BE},
\qquad
|\eta_z\rangle^{BE}
=
\frac{(\langle\psi_z|^A\otimes I_{BE})|\Psi\rangle^{ABE}}
{\sqrt{q_z}}.
\end{equation}
Consequently, $S(B)_{\eta_z}=S(E)_{\eta_z}$, and the private information becomes
\begin{equation}
\begin{aligned}
P
&=
S(B)_\rho-S(E)_\Psi
-
\sum_{z:q_z>0}q_z
\left[S(B)_{\eta_z}-S(E)_{\eta_z}\right]
\\
&=
S(B)_\rho-S(AB)_\rho
=
I(A\rangle B)_\rho.
\end{aligned}
\end{equation}
Thus, for any rank-one POVM on the entire system $A$, the private information equals the coherent information of the shared state. In particular, $P>0$ certifies positive one-way distillable entanglement.

In this case, we use the original state directly by identifying $A$ with the sensing register $X$, setting $\sigma_0^{XB}=\rho^{AB}$, and taking $C$ to be trivial. Let $H_X$ be a nondegenerate encoding Hamiltonian. Its eigenbasis need not coincide with the original POVM: the private information obtained by measuring in this basis also equals the coherent information.

Explicitly, let
\begin{equation}
|\Omega\rangle^{XBE}
=
\sum_{x\in X_+}
\sqrt{p_x}\,|x\rangle^X|\varphi_x\rangle^{BE},
\end{equation}
be a purification of $\sigma_0^{XB}$, where $\{|x\rangle\}$ is the encoding basis and $X_+$ contains its populated labels. Because no system $C$ is traced out, every conditional state $|\varphi_x\rangle^{BE}$ is pure. Hence
\begin{equation}
\begin{aligned}
P_X(\sigma_0)
&=
I(X;B)_{\omega_{\sigma_0}}
-
I(X;E)_{\omega_{\sigma_0}}
=
S(B)_{\omega_{\sigma_0}}
-
S(E)_{\omega_{\sigma_0}}
-
\sum_{x\in X_+}p_x
\left[
S(B)_{\varphi_x}
-
S(E)_{\varphi_x}
\right]
\\
&=
S(B)_{\sigma_0}
-
S(E)_{\Omega}
=
S(B)_{\sigma_0}
-
S(XB)_{\sigma_0}
=
I(X\rangle B)_{\sigma_0}.
\end{aligned}
\end{equation}
It follows that $P_X(\sigma_0)=P>0$. In this case, Alice holds only $X$, Bob holds $B$, and the Fisher-information lower bound is accessible by  LOCC without an assisting quantum channel.
When $C$ is nontrivial, transmitting $C$ is sufficient to make the joint $BC$ coherences certified by the private information accessible to Bob and thereby guarantee the LOCC Fisher information lower bound. By contrast, in the absence of $C$, the same type of lower bound is already accessible by LOCC without an assisting quantum channel.

To be rigorous, we independently prove the LOCC Fisher information lower bound for the case when $C$ is absent.
Applying Lemma \ref{lemm:all-pair-locc} with \(R=B\), and using
\begin{equation}
|h_j-h_k|\geq\delta_H(k-j),
\end{equation}
gives
\begin{equation}
F_{\mathrm{LOCC}}^{(X: B)}
\left(\sigma_{\theta}^{XB}\right)
\big|_{\theta=\theta_0}
\geq
\frac{\delta_H^2}{L_m}
\sum_{j<k}
(k-j)\frac{\|\sigma_{jk}^{B}\|_1^2}{p_j+p_k}.
\end{equation}
Using the same weighted Cauchy--Schwarz argument as in the proof of Theorem~\ref{thm:general-fixed-key-assisted},
\begin{equation}
F_{\mathrm{LOCC}}^{(X: B)}
\left(\sigma_{\theta}^{XB}\right)
\big|_{\theta=\theta_0}
\geq
\frac{\delta_H^2}{L_m^2}
\left(
\sum_{j<k}\|\sigma_{jk}^{B}\|_1
\right)^2.
\end{equation}
Lemma~\ref{lem:block-coherence-bound} and the Hermiticity of \(\sigma_0^{XB}\) give
\begin{equation}
\sum_{j<k}\|\sigma_{jk}^{B}\|_1
\geq
\frac{2^{H(X|E)_{\omega_{\sigma_0}}}-1}{2}.
\end{equation}
Using \(H(X|E)_{\omega_{\sigma_0}}\geq P_X(\sigma_0)\), we therefore obtain
\begin{equation}
F_{\mathrm{LOCC}}^{(X: B)}
\left(\sigma_{\theta}^{XB}\right)
\big|_{\theta=\theta_0}
\geq
\frac{\delta_H^2}{4L_m^2}
\left(2^{P_X(\sigma_0)}-1\right)^2.
\end{equation}
The right-hand side is strictly positive whenever $P_X(\sigma_0)>0$ and $\delta_H>0$.

\paragraph{Higher-rank measurements and coherent processing.}

If the POVM contains higher-rank elements, the conditional states $\tau_z^{BE}$ need not be pure, and the entropy cancellation used above no longer holds in general. Positive private information therefore need not imply positive coherent information or positive distillable entanglement. Nevertheless, we can construct a state of the form required by Theorem~\ref{thm:general-fixed-key-assisted} from $\rho^{AB}$.

Define the local isometry
\begin{equation}
V^{A\rightarrow XC}
=
\sum_x |x\rangle^X\otimes\sqrt{M_x}^{A\rightarrow C},
\end{equation}
where $C$ has the same Hilbert space as the original system $A$. The POVM completeness relation gives
\begin{equation}
V^\dagger V=\sum_x M_x^A=I_A.
\end{equation}
Thus, $V$ can be implemented by appending a pure local ancilla and applying a unitary. Applying this isometry to system $A$ at Alice's side produces
\begin{equation}
\sigma_0^{XBC}
=
(V^{A\rightarrow XC}\otimes I_B)
\rho^{AB}
(V^{A\rightarrow XC}\otimes I_B)^\dagger.
\end{equation}
Here, $X$ coherently records the measurement label, while $C$ retains the residual quantum system. 

To verify that the private information is preserved, consider the purification
\begin{equation}
|\widetilde{\Psi}\rangle^{XCBE}
=
(V^{A\rightarrow XC}\otimes I_{BE})|\Psi\rangle^{ABE}.
\end{equation}
where $|\Psi\rangle^{ABE}$ is the purification of $\rho^{AB}$, $|\widetilde{\Psi}\rangle^{XCBE}$ is the purification of $\sigma_0^{XBC}$. Dephasing $X$ in the basis $\{|x\rangle\}$ and tracing out $C$ gives
\begin{equation}
\omega_{\sigma_0}^{XBE}
=
\operatorname{Tr}_C
\Delta_X\!\left(
|\widetilde{\Psi}\rangle\langle\widetilde{\Psi}|^{XCBE}
\right)
=
\sum_x q_x |x\rangle\langle x|^X\otimes\tau_x^{BE}.
\end{equation}
This is exactly the classical--quantum state produced by the original state as in Eq. \ref{SI_eq:omega_ZBE}, with the classical register $Z$ relabeled as $X$. Both the outcome probabilities and the conditional states of Bob and Eve are therefore preserved, implying
\begin{equation}
P_X(\sigma_0)
=
I(X;B)_{\omega_{\sigma_0}}
-
I(X;E)_{\omega_{\sigma_0}}
=
P>0.
\end{equation}

The local isometry also preserves the coherent information across the bipartition between Alice and Bob
\begin{equation}
I(XC\rangle B)_{\sigma_0}=I(A\rangle B)_\rho.
\end{equation}
Thus, the construction can yield positive $P_X(\sigma_0)$ even when the coherent information across $XC:B$ is nonpositive. The distinction is that private information is evaluated using a measurement of the register $X$, with $C$ retained as a local auxiliary system, rather than a rank-one measurement of Alice's entire system.

The resulting state can be used for sensing by encoding the unknown parameter on $X$ through $U_\theta^X=e^{i\theta H_X}$. The local coherent processing can be chosen so that the output states $\ket{x}$ form the eigenbasis of $H_X$. Consequently, whenever the original state $\rho^{AB}$ yields positive private information under the chosen POVM, the resulting state $\sigma_0^{XBC}$ retains the same positive private information in the fixed encoding basis of $X$.
Since $P_X(\sigma_0)>0$, Theorem~\ref{thm:general-fixed-key-assisted} guarantees a positive lower bound on the assisted LOCC Fisher information when $C$ is transmitted through an appropriate assisting channel.
This construction provides an operational origin for the $XC$ decomposition: it can be created through local coherent processing and does not require the original system $A$ to possess this decomposition. 
The nondegenerate generator $H_X$ may act jointly on several degrees of freedom within $X$, and physical interactions involving $BC$ are also covered provided their combined action on the resource state is equivalent to the assumed unitary phase encoding on $X$.
}

\subsection{A coherent noisy-copy family with growing encoding-basis classical correlation}\label{SI_subsec:classical_correlation}

In this subsection, we construct a state family whose encoding-basis classical mutual information grows with dimension, while its global QFI remains finite and its assisted LOCC Fisher information vanishes. Let
\begin{equation}
    \dim X=\dim B=\dim C=n,
\end{equation}
we consider the state
\begin{equation}
\sigma_0^{XBC}=\beta\sum_b\ket{b,v_b}\bra{b,v_b}+a\sum_{b\neq b'}\ket{b,v_{b'}}\bra{b,v_{b'}}+ta\sum_{b\neq b'}\ket{b,v_{b'}}\bra{b',v_b}.
\end{equation}
where $b\in\{0,\ldots,n-1\}$, $\beta=\frac{u}{n}$, $a=\frac{1-u}{n(n-1)}$, $0<u,t<1$,
\begin{equation}
\begin{gathered}
\ket{b,v_b}^{XBC}=\ket{b}^X\ket{v_b}^{BC},\qquad\ket{v_b}^{BC}=\ket{b}^B\ket{\chi_b}^{C},\\
\ket{\chi_b}^{C}=\sum_{k=0}^{n-1}\sqrt{p_k^{(n)}}\,e^{2\pi i bk/n}\ket{k}^{C},     \qquad
    p_k^{(n)}>0,
    \qquad
    \sum_{k=0}^{n-1}p_k^{(n)}=1 .
\end{gathered}
\end{equation}
The coefficient matrix of
$\{\ket{\chi_b}\}_b$ is the product of the discrete Fourier matrix
and the invertible diagonal matrix
$\operatorname{diag}(\sqrt{p_0^{(n)}},\ldots,\sqrt{p_{n-1}^{(n)}})$.
Hence the states $\{\ket{\chi_b}\}_b$ are linearly independent and
span $C$. They are orthonormal only for the uniform distribution.
Irrespective of their overlaps on $C$,
$\langle v_b\vert v_{b'}\rangle=\delta_{bb'}$ because of the orthogonal $B$ components.

Let
\begin{equation}
D:=\sum_b\ket{b,v_b}\bra{b,v_b},\qquad S:=\sum_{b,b'}\ket{b,v_{b'}}\bra{b',v_b},\qquad I_{S}:=I^X\otimes\Pi^{BC}, \qquad \Pi^{BC}:=\sum_b\ket{v_b}\bra{v_b}^{BC}.
\end{equation}
The state can be rewritten as
\begin{equation}
\sigma_0=a(1-t)I_{ S}+at(I_{\mathcal S}+S)+\left[\beta-a(1+t)\right]D.
\end{equation}
We impose the parameter condition
\begin{equation}\label{SI_eq:beta_assumption}
\beta\geq a(1+t)
\qquad\Longleftrightarrow\qquad
u\geq\frac{1+t}{n+t}.    
\end{equation}
This ensures that the coefficient $\beta-a(1+t)$ is nonnegative. As shown below, this condition makes the background contribution in Eq.~\ref{SI_eq:p_alpha} nonnegative, allowing us to upper-bound the Fisher information obtainable from product measurements. The state is also separable across $X:BC$ under this condition.


The parameter $\theta$ is encoded according to
\begin{equation}
\sigma_\theta=(U_\theta^X\otimes I^{BC})\sigma_0(U_\theta^{X\dagger}\otimes I^{BC}),\qquad U_\theta^X=e^{i\theta H_X},\qquad H_X=\sum_bh_b\ket{b}\bra{b}.
\end{equation}
where the coefficients $h_{b}$ satisfy $\tr H_X=0$ and may be otherwise chosen arbitrarily.

\textit{Classical and private correlations.} We now calculate the classical correlation. After dephasing \(X\), we obtain
\begin{equation}
\omega^{XBC}=\Delta_X(\sigma_0^{XBC})=\frac{1}{n}\sum_b\ket{b}\bra{b}^X\otimes\gamma_b^{BC}.
\end{equation}
The label \(b\) is uniformly distributed, with \(p_b=1/n\), and the corresponding conditional state on \(BC\) is
\begin{equation}
\gamma_b^{BC}=u\ket{v_b}\bra{v_b}^{BC}+\frac{1-u}{n-1}\sum_{b'\neq b}\ket{v_{b'}}\bra{v_{b'}}^{BC}.
\end{equation}
Tracing out \(C\) gives
\begin{equation}
\gamma_b^B=u(\ket{b}\bra{b})^B+\frac{1-u}{n-1}\sum_{b'\neq b}(\ket{b'}\bra{b'})^B.
\end{equation}
The mutual information of this state is therefore
\begin{equation}
I(X;B)_\omega=\log_2n-h_2(u)-(1-u)\log_2(n-1).
\end{equation}
For fixed \(u\), as \(n\to\infty\),

\begin{equation}
I(X;B)_\omega=u\log_2n-h_2(u)+o(1),
\end{equation}
so the encoding basis classical mutual information can be arbitrarily large.

We now calculate the private information. Let

\begin{equation}
\ket{\Omega}^{XBCE}=\frac{1}{\sqrt{n}}\sum_b\ket{b}^{X}\ket{\Omega_b}^{BCE},\qquad \operatorname{Tr}_E\left[(\ket{\Omega_b}\bra{\Omega_b})^{BCE}\right]=\gamma_b^{BC},
\end{equation}
be a purification of \(\sigma_0\), and let \(\omega^{XBE}\) be obtained by dephasing \(X\) and tracing out \(C\). The relevant spectra give
\begin{equation}
S(\gamma_b^{BC})=h_2(u)+(1-u)\log_2(n-1),
\end{equation}
and
\begin{equation}
S(E)_\omega=S(\sigma_0)=h_2(u)+u\log_2n+(1-u)\left[\log_2\frac{n(n-1)}{2}+h_2\left(\frac{1+t}{2}\right)\right].
\end{equation}
Because the conditional state \(\ket{\Omega_b}^{BCE}\) is pure, \(S(E)_{\ket{\Omega_b}}=S(\gamma_b^{BC})\), and therefore
\begin{equation}
H(X|E)_\omega=\log_2n+S(\gamma_b^{BC})-S(\sigma_0)=(1-u)L(t),\qquad L(t):=1-h_2\left(\frac{1+t}{2}\right).
\end{equation}
On the other hand,
\begin{equation}
H(X|B)_\omega=h_2(u)+(1-u)\log_2(n-1).
\end{equation}
Consequently, we can get
\begin{equation}
P_X(\sigma_0):=I(X;B)_\omega-I(X;E)_\omega=(1-u)L(t)-h_2(u)-(1-u)\log_2(n-1).
\end{equation}
For \(n\geq3\), we have \(\log_2(n-1)\geq1>L(t)\), and hence
\begin{equation}
P_X(\sigma_0)=(1-u)\left[L(t)-\log_2(n-1)\right]-h_2(u)<0,\qquad n\geq3.
\end{equation}

\textit{Global QFI.}
For every unordered pair \(b<b'\), define
\begin{equation}
\ket{\psi_{bb'}^\pm}=\frac{\ket{b,v_{b'}}\pm\ket{b',v_b}}{\sqrt{2}}.
\end{equation}
These states have eigenvalues \(a(1\pm t)\), and
\begin{equation}
\bra{\psi_{bb'}^+}(H_X\otimes I^{BC})\ket{\psi_{bb'}^-}=\frac{h_b-h_{b'}}{2}.
\end{equation}
The spectral formula for the QFI therefore gives \cite{paris2009quantum,BraunsteinCaves1994}
\begin{equation}
F_Q(\sigma_\theta)=2at^2\sum_{b<b'}(h_b-h_{b'})^2=\frac{2t^2(1-u)}{n-1}\operatorname{Tr}H_X^2,
\end{equation}
where we used \(\sum_{b<b'}(h_b-h_{b'})^2=n\operatorname{Tr}H_X^2\) and \(\operatorname{Tr}H_X=0\).

\textit{Upper bound on the unassisted LOCC Fisher information.} For the unassisted setting, Alice holds $XC$ and Bob holds $B$. Define
\begin{align}
S_\theta
=
(U_\theta^X\otimes I^{BC})S
(U_\theta^{X\dagger}\otimes I^{BC}),
\qquad
\Delta_H=\max_{b,b'}|h_b-h_{b'}|.
\end{align}
The decomposition above and $\beta\geq a(1+t)$ imply
\begin{align}
\sigma_\theta\geq a(1-t)I_S,
\qquad
\dot{\sigma}_\theta=at\dot S_\theta.
\label{eq:unassisted-basic-bound}
\end{align}
Moreover, $\dot S_\theta$ is supported on $I_S$ and has eigenvalues
$\pm|h_b-h_{b'}|$ on the subspace spanned by
$\{|b,v_{b'}\rangle,|b',v_b\rangle\}$. Therefore,
\begin{align}
\|\dot S_\theta\|_\infty=\Delta_H.
\end{align}

Let $\{M_\alpha\}_\alpha$ be an arbitrary LOCC POVM across $XC:B$, and define
\begin{equation}
p_\alpha=\operatorname{Tr}(M_\alpha\sigma_\theta),
\qquad
s_\alpha=\operatorname{Tr}(M_\alpha I_S),
\qquad
d_\alpha=\operatorname{Tr}(M_\alpha\dot S_\theta).
\end{equation}
We have
\begin{align}
p_\alpha\geq a(1-t)s_\alpha,
\qquad
\dot p_\alpha=atd_\alpha,
\qquad
|d_\alpha|\leq\Delta_Hs_\alpha.
\end{align}
Consequently,
\begin{align}
\frac{\dot p_\alpha^2}{p_\alpha}
\leq
\frac{at^2\Delta_H}{1-t}|d_\alpha|.
\label{eq:unassisted-outcome-bound}
\end{align}

Every LOCC effect is separable and hence PPT across $XC:B$. Thus,
$\{M_\alpha^{T_B}\}_\alpha$ is also a POVM. We note the following inequality for any Hermitian matrix $A$,
\begin{align}
\sum_\alpha
\left|
\operatorname{Tr}(E_\alpha A)
\right|
\leq
\|A\|_1.
\end{align}
This follows by decomposing $A$ into its positive and negative parts and using the positivity of each $E_\alpha$. Equality is attained by choosing the POVM elements $E_\alpha$ to be the projections onto the positive and negative spectral subspaces of $A$.
Thus, we obtain
\begin{align}
\sum_\alpha|d_\alpha|
&=
\sum_\alpha
\left|
\operatorname{Tr}
\left(M_\alpha^{T_B}\dot S_\theta^{T_B}\right)
\right|
\leq
\|\dot S_\theta^{T_B}\|_1.
\label{eq:unassisted-ppt-bound}
\end{align}

We now try to bound $\|\dot S_\theta^{T_B}\|_1$. Define
\begin{equation}
    \overline{\rho}^{C}
    =
    \frac{1}{n}\sum_{b=0}^{n-1}
    \ket{\chi_b}\!\bra{\chi_b}
    =
    \sum_{k=0}^{n-1}
    p_k^{(n)}\ket{k}\!\bra{k},
    \qquad
    R_n
    =
    \left(\Tr\sqrt{\overline{\rho}^{C}}\right)^2
    =
    \left(
        \sum_{k=0}^{n-1}\sqrt{p_k^{(n)}}
    \right)^2 .
\end{equation}
Let
\begin{equation}
    \ket{e_b}^{XB}=\ket{b}^{X}\ket{b}^{B},
    \qquad
    \ket{\xi_\theta}^{XBC}
    =
    \sum_{b=0}^{n-1}
    e^{-i\theta h_b}\ket{e_b}^{XB}\ket{\chi_b}^{C},
\end{equation}
and define $K=XB$. A direct calculation gives
\begin{equation}
    S_\theta^{T_B}
    =
    \left(
        \ket{\xi_\theta}\!\bra{\xi_\theta}
    \right)^{T_K},
\end{equation}
where $T_K$ denotes the partial transpose on  $K$ 
with respect to the basis $\{\ket{e_b}\}_b$.

The reduced operator of $\ket{\xi_\theta}\bra{\xi_\theta}$ on $C$ is $n\overline{\rho}^{C}$. Writing the Schmidt decomposition as $\ket{\xi_\theta}^{XBC}=\sum_j s_j\ket{a_j}^{K}\ket{c_j}^C$, we find that $S_\theta^{T_B}=\sum_{i,j}s_is_j\ket{a_j}\bra{a_i}^K\otimes\ket{c_i}\bra{c_j}^C$ has eigenvalues $s_i^2,\pm s_{i}s_j$, where the latter pair occurs for each $i<j$. Hence,
\begin{equation}
    \left\|S_\theta^{T_B}\right\|_1
    =\sum_i s_i^2+\sum_{i<j}2s_is_j=(\sum_i s_i)^2=
    \left(
        \Tr\sqrt{n\overline{\rho}^{C}}
    \right)^2
    =
    nR_n .
\end{equation}
Define
\[
    H_K:=\sum_{b=0}^{n-1}h_b\ket{e_b}\!\bra{e_b}.
\]
The encoded operator satisfies
\begin{equation}
    S_\theta^{T_B}
    =
    e^{i\theta H_K}
    S_0^{T_B}
    e^{-i\theta H_K}.
\end{equation}
Consequently,
\begin{equation}
\begin{aligned}
    \left\|\dot S_\theta^{T_B}\right\|_1
    &=
    \left\|
        i[H_K,S_\theta^{{T_B}}]
    \right\|_1 =      \left\|
        i[H_K-cI_K,S_\theta^{{T_B}}]
    \right\|_1                                          \\       
    &\leq
    2\min_{c\in\mathbb{R}}
    \left\|H_K-cI_K\right\|_\infty
    \left\|S_\theta^{T_B}\right\|_1                                  
    =
    \Delta_H nR_n .
\end{aligned}
\end{equation}
The second equality holds for every $c\in\mathbb{R}$ because the identity commutes with $S_\theta^{T_B}$. Writing $h_{\max}=\max_b h_b$ and $h_{\min}=\min_b h_b$, the minimum is attained at $c=(h_{\max}+h_{\min})/2$, for which $\|H_K-cI_K\|_\infty=(h_{\max}-h_{\min})/2=\Delta_H/2$.

Combining the above results, we obtain
\begin{equation}
\begin{aligned}
    F_{\mathrm{LOCC}}^{(XC:B)}(\sigma_\theta)
    &=
    \sum_\alpha\frac{\dot p_\alpha^2}{p_\alpha}      \leq
    \frac{at^2\Delta_H}{1-t}
    \sum_\alpha |d_\alpha|                     
   \leq
    \frac{at^2\Delta_H}{1-t}
    \left\|\dot S_\theta^{T_B}\right\|_1                              \\
    &\leq
    \frac{at^2\Delta_H^2nR_n}{1-t}             
    =
    \frac{t^2(1-u)\Delta_H^2R_n}
         {(1-t)(n-1)} .
\end{aligned}
\end{equation}
For the choice
\[
    h_b=\frac{b-(n-1)/2}{n-1},
\]
we have $\Delta_H=1$, and hence
\begin{equation}
    F_{\mathrm{LOCC}}^{(XC:B)}(\sigma_\theta)
    \leq
    \frac{t^2(1-u)R_n}{(1-t)(n-1)}
    =
    O\!\left(\frac{R_n}{n}\right).
\end{equation}
Therefore, the unassisted LOCC Fisher information is $O(1/n)$ whenever $R_n$ remains bounded independently of $n$.



Note that
\begin{equation}
R_n=2^{H_{1/2}(p^{(n)})},    
\end{equation}
where $H_{1/2}$ denotes the Rényi entropy of order $1/2$ and hence 
\begin{equation}
    1
    \leq
    R_n
    =
    \left(
        \sum_{k=0}^{n-1}\sqrt{p_k^{(n)}}
    \right)^2
    \leq n .
\end{equation}
Consequently,  $R_n$ can remain bounded independently of $n$ for proper choice of $\{p_k^{(n)}\}_k$. Any $p_k^{(n)}$ satisfying
$R_n=O(1)$ yields
\begin{equation}
F_{\mathrm{LOCC}}^{(XC:B)}=O(1/n).    
\end{equation}

We now consider a few examples of the choice of $\{p_k^{(n)}\}_k$. The uniform choice
\[
    p_k^{(n)}=\frac{1}{n}
\]
gives orthonormal tags and $R_n=n$. In this case, the present
unassisted upper bound is only $O(1)$ and does not establish a
vanishing LOCC Fisher information.

For the choice
\begin{equation}
    p_k^{(n)}
    =
    \frac{(1-\kappa)\kappa^k}{1-\kappa^n},
    \qquad
    0<\kappa<1.    
\end{equation}
We find
\begin{equation}
\begin{aligned}
    R_n
    &=
    \frac{1+\sqrt{\kappa}}{1-\sqrt{\kappa}}
    \frac{1-\kappa^{n/2}}{1+\kappa^{n/2}}                            
    &\leq
    \frac{1+\sqrt{\kappa}}{1-\sqrt{\kappa}}
    =
    O(1).
\end{aligned}
\end{equation}

A simpler choice is
\[
\begin{aligned}
    p_0^{(n)}
    &=
    (1-q)\left(1-\frac{1}{n}\right),                                \\
    p_1^{(n)}
    &=
    q\left(1-\frac{1}{n}\right),                                    \\
    p_k^{(n)}
    &=
    \frac{1}{n(n-2)},
    \qquad
    2\leq k\leq n-1,
\end{aligned}
\qquad
0<q<1.
\]
For this distribution,
\[
\begin{aligned}
    R_n
    &=
    \left[
        \sqrt{(1-q)\left(1-\frac{1}{n}\right)}
        +
        \sqrt{q\left(1-\frac{1}{n}\right)}
        +
        \sqrt{\frac{n-2}{n}}
    \right]^2                                                       \\
    &\longrightarrow
    \left(1+\sqrt{q}+\sqrt{1-q}\right)^2
    =
    O(1).
\end{aligned}
\]

\textit{Upper bound on the assisted LOCC Fisher information.} Every LOCC measurement across $X:BC$ is separable, and every separable POVM can be refined into rank-one product effects. Since classical Fisher information cannot decrease under outcome refinement, it suffices to consider rank-one product effects across \(X:BC\) of the form
\begin{equation}
M_\alpha=w_\alpha\ket{p_\alpha}\bra{p_\alpha}^X\otimes\ket{q_\alpha}\bra{q_\alpha}^{BC}.
\end{equation}
Define
\begin{equation}
\begin{aligned}
    &s_\alpha:=\bra{q_\alpha}\Pi^{BC}\ket{q_\alpha},\qquad z_\alpha:=\bra{p_\alpha}U_\theta^XW^{BC\to X}\ket{q_\alpha},\\
    & y_\alpha:=\bra{p_\alpha}H_XU_\theta^XW^{BC\to X}\ket{q_\alpha},\qquad r_\alpha:=\sum_b\left|\bra{p_\alpha}\ket{b}\right|^2\left|\bra{q_\alpha}\ket{v_b}\right|^2.
\end{aligned}
\end{equation}
where $W^{BC\to X}=\sum_b\ket{b}^X\bra{v_b}^{BC}$.
The corresponding outcome probability is
\begin{equation}\label{SI_eq:p_alpha}
p_\alpha(\theta)=w_\alpha\left[as_\alpha+ta|z_\alpha|^2+\left(\beta-a-ta\right)r_\alpha\right].
\end{equation}
The Cauchy--Schwarz inequality gives \(\lvert z_\alpha\rvert^2\leq s_\alpha\leq1\). Moreover,
\begin{equation}
|z_\alpha|^2\leq s_\alpha\leq1,\qquad \left|\frac{d}{d\theta}|z_\alpha|^2\right|^2\leq4|z_\alpha|^2|y_\alpha|^2.
\end{equation}
Using the assumption in Eq.~\ref{SI_eq:beta_assumption}, \(\beta-a-ta\geq0\) and \(\lvert z_\alpha\rvert^2\leq s_\alpha\), we obtain
\begin{equation}
\frac{\dot p_\alpha(\theta)^2}{p_\alpha(\theta)}\leq\frac{4w_\alpha t^2a}{1+t}|y_\alpha|^2.
\end{equation}
Completeness of the POVM, together with \(WW^\dagger=I^X\), implies
\begin{equation}
\sum_\alpha w_\alpha\left|\bra{p_\alpha}H_XU_\theta^XW^{BC\to X}\ket{q_\alpha}\right|^2=\operatorname{Tr}H_X^2.
\end{equation}
It follows that
\begin{equation}
F_{\mathrm{SEP}}^{(X:BC)}(\sigma_\theta)\leq\frac{4t^2(1-u)}{(1+t)n(n-1)}\operatorname{Tr}H_X^2=\frac{2}{n(1+t)}F_Q(\sigma_\theta).
\end{equation}
Let \(\mathcal{N}^{C\to C'}\) be any prescribed parameter-independent channel acting directly on \(C\). Then
\begin{equation}
F_{\mathrm{LOCC}}^{(X:BC')}\left[\left(\operatorname{id}^{XB}\otimes\mathcal{N}^{C\to C'}\right)(\sigma_\theta)\right]\leq F_{\mathrm{SEP}}^{(X:BC)}(\sigma_\theta)\leq\frac{2}{n(1+t)}F_Q(\sigma_\theta).
\end{equation}
We have thus shown that even arbitrarily large encoding-basis classical mutual information does not guarantee a nonvanishing assisted LOCC Fisher information.

As an example, we choose
\begin{equation}
h_b=\frac{b-(n-1)/2}{n-1}.
\end{equation}
Thus,
\begin{equation}
\operatorname{Tr}H_X=0,\qquad \operatorname{Tr}H_X^2=\frac{n(n+1)}{12(n-1)}.
\end{equation}
Consequently,
\begin{equation}\begin{aligned}
&F_Q(\sigma_\theta)=\frac{t^2(1-u)n(n+1)}{6(n-1)^2}\longrightarrow\frac{t^2(1-u)}{6}=\Theta(1),\\
&F_{\mathrm{LOCC}}^{(X:BC')}\left[\left(\operatorname{id}^{XB}\otimes\mathcal{N}^{C\to C'}\right)(\sigma_\theta)\right]\leq\frac{t^2(1-u)(n+1)}{3(1+t)(n-1)^2}=O(1/n).    
\end{aligned}
\end{equation}

The advantage of the global measurement in this example arises from at least two distinct mechanisms.

First, the phase information is encoded in the coherences \(\ket{b,v_{b'}}\bra{b',v_b}\), for each of the \(n(n-1)/2\) unordered pairs \(b<b'\). A global measurement can resolve all these coherences in a single measurement setting because each pair spans an orthogonal two-dimensional subspace. In particular, one can measure in the Bell-like basis \(\ket{\Psi_{bb'}^\pm}=(\ket{b,v_{b'}}\pm\ket{b',v_b})/\sqrt{2}\). States associated with distinct pairs satisfy \(\bra{\Psi_{bb'}^{s}}\ket{\Psi_{\tilde b\tilde b'}^{s'}}=\delta_{b\tilde b}\delta_{b'\tilde b'}\delta_{ss'}\) for \(b<b'\) and \(\tilde b<\tilde b'\), while the two states within each pair are also orthogonal. The Fisher-information contributions from all pairs can therefore be collected simultaneously. By contrast, an LOCC measurement accesses the coherence of a given pair through correlations between local superposition measurements. For example, a pairwise readout can be constructed using product states such as \(\ket{\phi_{bb',\pm}}^X\otimes\ket{\varphi_{bb',\pm}}^{BC}\), where \(\ket{\phi_{bb',\pm}}^X=(\ket{b}^X\pm\ket{b'}^X)/\sqrt{2}\) and \(\ket{\varphi_{bb',\pm}}^{BC}=(\ket{v_b}^{BC}\pm\ket{v_{b'}}^{BC})/\sqrt{2}\). The crucial difference is that local superposition states associated with different pairs are generally nonorthogonal whenever the pairs share an index. Consequently, these pairwise readouts cannot all be realized as independent outcomes of a single local measurement basis. A pair-resolving local setting can cover only on the order of \(n\) disjoint pairs, whereas there are on the order of \(n^2\) pairs in total, so covering all pairs requires on the order of \(n\) incompatible settings. This incompatibility provides an intuitive explanation for the factor-of-order-\(n\) loss in the LOCC Fisher information relative to the global QFI.

Second, when we perform the global measurement in the Bell-like basis \(\ket{\Psi_{bb'}^\pm}=(\ket{b,v_{b'}}\pm\ket{b',v_b})/\sqrt{2}\), the first term \(\beta\sum_b\ket{b,v_b}\bra{b,v_b}\) does not contribute to the outcome probabilities. In contrast, for an LOCC measurement with product vectors such as \(\ket{\phi_{bb',\pm}}^X\otimes\ket{\varphi_{bb',\pm}}^{BC}\), the first term \(\beta\sum_b\ket{b,v_b}\bra{b,v_b}\) enters the outcome probabilities as a noise background. Since \(\beta/a=\Theta(n)\), this background reduces the Fisher information of the LOCC measurement by a factor of \(a/(a+\beta)=\Theta(1/n)\). Thus, this background-noise effect provides a second reason why global measurements can exhibit an advantage.

We now give an intuitive interpretation of the construction of \(\sigma_0^{XBC}\). The phase information is carried by the exchange coherences \(ta\sum_{b\neq b'}\ket{b,v_{b'}}\bra{b',v_b}\), while the diagonal part \(a\sum_{b\neq b'}\ket{b,v_{b'}}\bra{b,v_{b'}}\) is required for positivity.   We  add the noise term \(\beta\sum_b\ket{b,v_b}\bra{b,v_b}\) and choose \(\beta\geq a(1+t)\), which makes the background contribution nonnegative in Eq.~\ref{SI_eq:p_alpha}. Importantly, this term is not uncorrelated noise, but a phase-insensitive component that is perfectly classically correlated across \(X:BC\). Moreover, \(H(X|B)_\omega=h_2(u)+(1-u)\log_2(n-1)\) grows with \(n\), whereas \(H(X|E)_\omega=(1-u)L(t)\) remains bounded, ensuring that \(P_X(\sigma_0)<0\) for \(n\geq3\). 
The resulting state therefore has large encoding-basis classical mutual information and finite global QFI, has negative fixed-\(X\) private information, and, for choices satisfying \(R_n=O(1)\), both its unassisted and assisted LOCC Fisher informations vanish as \(n\to\infty\).

\yk{
\subsection{Example of positive LOCC Fisher information with negative private information}

Positive private information is sufficient but not necessary for nonzero assisted LOCC Fisher information. To illustrate this, consider the two-qubit state
\begin{equation}
\rho_0^{XB}=\frac{1}{4}\left(I^{XB}+r\sigma_x^X\otimes\sigma_x^B\right),\qquad 0<r\leq1,
\end{equation}
where $\sigma_x$ denotes the Pauli operator and the subsystem $C$ is taken to be trivial. The unknown parameter $\theta$ is encoded through $U_\theta^X=\exp(i\theta|1\rangle\langle1|^X)$, giving
\begin{equation}
\rho_\theta^{XB}
=U_\theta\rho_0^{XB} U_\theta^\dagger.
\end{equation}
Both reduced states remain maximally mixed, so the parameter dependence resides entirely in the correlations.

We first calculate the private information in the encoding basis $\{\ket{0},\ket{1}\}$. For a purification $\ket{\Omega}^{XBE}$ of $\rho_0^{XB}$, dephasing $X$ in this basis gives
\begin{equation}
\omega^{XBE}=\Delta_X(\ket{\Omega}\bra{\Omega}^{XBE})
\end{equation}
The corresponding private information is
\begin{equation}
P_X(\rho_0)
=I(X;B)_\omega-I(X;E)_\omega
=-h_2\left(\frac{1+r}{2}\right)\leq0.
\end{equation}
where $h_2(p)$ is the binary entropy. The inequality is strict for $0<r<1$.

Nevertheless, the state has nonzero LOCC Fisher information. Suppose Alice and Bob both measure in the basis $|\pm\rangle=(|0\rangle\pm|1\rangle)/\sqrt{2}$. Their joint outcome probabilities are
\begin{equation}
p(a,b)
=\frac{1}{4}\left(1+ab\,r\cos\theta\right),
\qquad a,b\in\{+1,-1\}.
\end{equation}
The resulting LOCC Fisher information is
\begin{equation}
F
=\frac{r^2\sin^2\theta}{1-r^2\cos^2\theta}.
\end{equation}
which is positive for $0<r<1$. This measurement remains available if an uncorrelated auxiliary system $C$ is appended and transmitted through an assisting channel. Thus, nonpositive private information is compatible with nonzero assisted LOCC Fisher information.
However, the example in Sec. \ref{SI_subsec:classical_correlation} shows that, when private information is nonpositive, the LOCC Fisher information can also vanish as the system dimension increases. Therefore, when the private information is nonpositive, we should not expect a guaranteed nonzero lower bound on the LOCC Fisher information.

}

\yk{
\subsection{Possible multipartite extensions}

Multipartite private correlations have been extensively studied in quantum communication \cite{hillery1999quantum,salek2025new,augusiak2009multipartite}. The multipartite setting admits different partitions of the parties and different requirements for recovering the secret. For example, treating all recipients as a joint receiver gives a bipartite privacy criterion across that partition, whereas conference-key agreement requires every recipient to recover the same secret key. Secret-sharing schemes can instead require cooperation within an authorized group, while keeping the secret inaccessible to unauthorized groups, and an external adversary. These scenarios lead to different operational notions of private correlations.

The connection between privacy and sensing established here motivates a systematic investigation of the corresponding bounds on multipartite LOCC Fisher information. Such an investigation could clarify how different partitions and access requirements in private communication translate into measurement restrictions and assistance structures in quantum sensing. It could also establish whether and how the quantum resources supporting different private communication tasks yield distinct sensing guarantees.

A central issue is the arrangement of assisting channels. In the bipartite setting considered here, the derivation of our assisted lower bound relies on Alice transmitting the unencoded subsystem to Bob through an assisting channel. In a multipartite setting, determining which parties should transmit their unencoded subsystems and which parties should receive them is already nontrivial. Identifying appropriate assistance structures and deriving the resulting sensing guarantees under multipartite LOCC therefore remain important directions for future work.

}

\yk{
\subsection{Relation to quantum error correction code}

Since our work concerns noisy quantum sensing, it is natural to ask how our lower bound relates to quantum error correction. Quantum error correction is compatible with our framework, provided that the bound is evaluated for the appropriate resource state and the assumptions of the theorem remain satisfied as detailed below.

If a quantum error-correcting code is used during preparation, distribution, or storage, the state incorporating this code serves as the starting state $\sigma_\theta^{XBC}$. Syndrome measurements and recovery operations need not have been completed beforehand. Theorem~\ref{thm:general-fixed-key-assisted} does not impose a particular internal structure on the resource state beyond its stated assumptions and therefore accommodates states incorporating quantum error correction codes. The private information should be evaluated for the state prepared using the code, rather than for the original state before the code was introduced. Positive private information then certifies the assisted sensing guarantee.

Local error correction performed during readout is already included in the optimization over LOCC measurements. This optimization allows arbitrary local quantum operations, including the introduction of local ancillas, syndrome measurements, and conditional recovery operations. Any parameter-independent local recovery followed by an LOCC measurement constitutes an LOCC measurement on the original encoded state. Consequently, allowing such recovery operations does not invalidate the bounds for the specified state families.

Error correction can also be interleaved with parameter encoding, as in the sensing protocol assisted by quantum error correction \cite{zhou2018achieving}. Theorem~\ref{thm:general-fixed-key-assisted}  applies to such protocols when the state obtained after the combined error-correction and parameter-encoding process  serves as the input to the assisted sensing protocol. The private information is evaluated for this state in the eigenbasis of the effective encoding generator.
}



\end{document}